\documentclass{sig-alternate}
\usepackage{jeffstyle}
\usepackage{epstopdf}
\newlength{\figsize} 
\newlength{\widthfigsize} 
\usepackage[numbers]{natbib}
\begin{document}
\conferenceinfo{KDD}{2015 Sydney, Australia}

\title{Binary Coding in Stream}

\numberofauthors{2} 
\author{
\alignauthor
Mina Ghashami
       \\
       \affaddr{University of Utah}\\
       \affaddr{Utah, USA}\\
       \email{ghashami@cs.utah.edu}
\alignauthor
Amirali Abdullah
\\
       \affaddr{University of Utah}\\
       \affaddr{Utah, USA}\\
       \email{amirali@cs.utah.edu}
}

\date{20 Feb 2015}

\maketitle
\begin{abstract}
Big data is becoming ever more ubiquitous, ranging over massive video repositories, document corpuses, image sets and Internet routing history. Proximity search and clustering are two algorithmic primitives fundamental to data analysis, but suffer from the ``curse of dimensionality'' on these gigantic datasets. 
A popular attack for this problem is to convert object representations into short binary codewords, while approximately preserving near neighbor structure. However, there has been limited research on constructing codewords in the ``streaming" or ``online" settings often applicable to this scale of data, where one may only make a single pass over data too massive to fit in local memory. 

In this paper, we apply recent advances in matrix sketching techniques to construct binary codewords in both streaming and online setting. 
Our experimental results compete outperform several of the most popularly used algorithms, and we prove theoretical guarantees on performance in the streaming setting under mild assumptions on the data and randomness of the training set.
\end{abstract}
%
%

\section{Introduction}
Due to overwhelming increase in sheer volume of data being generated every day, fundamental algorithmic primitives of data analysis are being run on ever larger data sets. These primitives include approximating nearest neighbour search~\cite{indyk1998approximate, binarylsh}, clustering~\cite{ball1967clustering,kanungo2002efficient}, low dimensional embeddings~\cite{niyogi2004locality,belkin2001laplacian}, or learning distributions from a limited number of samples~\cite{learnability} etc.

A prominent approach for handling gigantic datasets is to convert object representations to short \textit{binary codewords} such that similar objects map to similar binary codes.
%
Binary representation is widely used in data analysis tasks, for example Song et.al~\cite{videosearch} gave an algorithm for converting a large video dataset into a set of binary hashes. Seo~\cite{seo} proposed a binary hashing scheme for music retrieval. Fergus, Weiss and Torralba~\cite{fergus2009semi} employed a spectral hashing scheme for labeling gigantic image datasets in semi-supervised setting. Julie and Triggs~\cite{recognition} used binary feature vectors for visual recognition of objects inside images. Guruswami and Sahai~\cite{ecc} give an embedding into Hamming space that reduces multi-class learning to an easier binary classification problem.

Codewords as succinct representation of data serve multiple purposes: 1) They can be used for dimensionality reduction, 2) They can emphasize user-desired distance thresholds, i.e. to encode data points such that near neighbors become much closer in Hamming space, rather than a simple proportionate embedding of distances, and 3) They allow the use of efficient tree based search data structures and enable the use of nearest neighbor techniques in Hamming space. (For more on how to conduct such searches quickly in Hamming space, see for instance the work by Norouzi, Punjani and Fleet~\cite{fasthamming} or by 
Esmaeili, Ward and Fatourechi~\cite{2012fast}.)



Sometimes these codes may be found trivially, e.g. if a dataset is already described by binary features, or is partitioned in a locality preserving and hierarchical manner. 
However where we are not so fortunate, we need to \emph{learn} them by seeking help of a constructive similarity function.
For instance,
unsupervised methods derive codewords from feature vectors in Euclidean space, or construct them from a data independent affinity matrix. 
On the opposite side, supervised methods \cite{torralba2008small,semantic} take additional contextual information into account and use a similarity notion that is semantically meaningful for codewords, e.g. two documents are similar if they are about the same topic or two images are similar if they contain same objects and colors.

On a meta-level, any binary coding scheme should satisfy three following properties to be considered effective: 
\begin{enumerate}
\item{The codes should be short so that we can store large datasets in memory.}
\item{Codes should be similarity-preserving; i.e., similar data points should map to similar binary codes while far data points should not collapse to small neighborhoods.}
\item{The learning algorithm should efficiently compute codes for newly inserted points without having to recompute the entire codebook.}
\end{enumerate}
The need to simultaneously satisfy all three constraints above makes learning binary codes a challenging problem.

Broadly speaking, binary coding techniques fall into two categories: first class is the family of techniques, referred to as \textit{symmetric}, which binarize both datapoints of a dataset or database and query points, usually according to the same hashing scheme. This class includes locality sensitive hashing(LSH)~\cite{indyk1998approximate}, spectral hashing~\cite{weiss2009spectral}, 
locality sensitive binary codes\cite{raginsky2009locality}, Iterative Quantization(ITQ)~\cite{gong2011iterative} or semi-supervised hashing\cite{wang2012semi} techniques.
In contrast, the second class of methods, namely \emph{asymmetric} algorithms, binarize only data points and not query points, e.g. ~\cite{jegou2011product,dong2008asymmetric,gordo2014asymmetric,jegou2010aggregating}. These methods achieve higher accuracy due to greater precision in the query description, yet still have the storage and efficiency gains from binarizing the ground dataset.

\section{Background and Notation}
First, we briefly review some notation. We use lower case letters to denote functions, e.g. $w(x)$ and upper case letters to represent matrices, e.g. $W$. An $n \times d$ matrix $A$ can be written as a set of $n$ rows as $[A_{1,:}; A_{2,:}; \ldots, A_{n,:}]$ where each row $A_{i.:}$ is a datapoint of length $d$. Equivalently, this matrix can be written as a set of $d$ columns as $[A_{:,1}, A_{:,2}, \ldots, A_{:,d}]$. The element at row $i$ and column $j$ of matrix $A$ is denoted by $A_{ij}$.
The Frobenius norm of a matrix $A$ is defined $\|A\|_F = \sqrt{\sum_{i=1}^n \|A_{i,:}\|^2}$ where $\|A_{i,:}\|$ is Euclidean norm of $A_{i,:}$.  
Let $A_k$ refer to the best rank $k$ approximation of $A$, specifically $A_k = {\argmax}_{C : \rank(C) \leq k} \|A - C\|_F$.  
The singular value decomposition of $A \in \R^{n \times d}$, written $\svd(A)$, produces three matrices $[U,\Sigma,V]$ so that $A = U \Sigma V^T$.  Matrices $U \in \R^{n \times n}$ and $V \in \R^{d \times d}$ are orthogonal and their columns are the left singular vectors and right singular vectors, respectively. Matrix $\Sigma \in \R^{n \times d}$ is all $0$s except for the diagonal entries $\{\Sigma_{1,1}, \Sigma_{2,2}, \ldots, \Sigma_{r,r}\}$, the \emph{singular values}, where $r \leq d$ is the rank.  Note that $\Sigma_{j,j} \geq \Sigma_{j+1,j+1}$ for all $1 \leq j \leq r-1$, spectral norm of a matrix is $\|A\|_2 = \Sigma_{1,1}$, and $\Sigma_{j,j} = \|A V_{:,j}\|$ describes the norm along direction $V_{:,j}$.  
Numeric rank of a matrix $A$ is defined as $\|A\|_F^2/\|A\|_2^2$ and trace of a square matrix $M \in \R^{n \times n}$ is $\Tr(M) = \sum_{i=1}^n M_{i,i}$
For square matrix $A \in \R^{n \times n}$, eigen decomposition of $A$ is $\eig(A) = U \Lambda U^T$ where $U \in \R^{n \times n}$ contains eigen vectors as columns, and $\Lambda \in \R^{n \times n}$ is a diagonal matrix containing eigen values $\{\Lambda_{1,1}, \Lambda_{2,2}, \ldots, \Lambda_{n,n}\}$ in non-increasing order. 
Finally, expected value of a matrix is defined as the matrix of expected values, i.e. 
\[
\E[A] = \left( \begin{array}{ccc} \E[A_{1,1}] & \cdots & \E[A_{1,d}] \\ \vdots & \vdots & \vdots \\ \E[A_{n,1}] & \cdots & \E[A_{n,d}] \end{array} \right)
\]
\subsection{Related Works}
One of the basic and most popular binary encoding schemes is ``Locality Sensitive Hashing'' (LSH)\cite{datar2004locality} which uses random projections to embed data into lower dimensional space. This is done by employing a class of functions called locality-sensitive hash functions under which similar points collide with high probability.
A family of hash functions $H$ is called $(r,cr,P_1,P_2)$-sensitive if for any two points $p, q \in \R^d$ and any hash function $h \in H$, the following two properties hold:
\begin{enumerate}
\item{If $\|p - q\| \leq r$ then $P_{H}[h(p) = h(q)] \geq P_1$ and 2) if $\|p - q\| \geq cr$ then $P_{H}[h(p) = h(q)] \leq P_2$, where $P_H$ denotes the probability of an event under family of hash functions $H$, and $h(p)$ is the hashed value of point $p$ under hash function $h$. Note in this definition $r >0 $ is a threshold on distance and $c$ is an approximation ratio, and in order for LSH family to be useful it should be that $P_1 > P_2$.}
\item{LSH is a data independent method and can be done in streaming setting as it does not need to store data points and hashing or projection can be done on the fly. It is folklore that for random datasets LSH is near optimal, but in practice is generally outperformed by methods that use spectrum of data. To simplify somewhat, $k$ bit binary codewords of LSH can be assigned to a point in $\R^d$ by taking dot product with a collection of $k$ random vectors, and assigning each bit as $0$ or $1$ according to the sign of the value obtained~\cite{binarylsh}.}
\end{enumerate} 

One of the most famous binary encoding schemes is ``Spectral Hashing''(SH) \cite{weiss2009spectral}. If $W \in \R^{n \times n}$ is similarity matrix and $Y \in \R^{n \times k}$ is the binary coding matrix for $k$ being the length of codewords, then this method 
formulates the problem as minimizing $\sum_{i,j} W_{i,j} \| Y_{i,:} - Y_{j,:} \|^2$ with subject to $Y(i,j) \in \{-1, 1 \}$, balance constraint, i.e. $\sum_{j=1}^k Y_{i,j} = 0$ for each binary codeword $Y_{i,:}$, and evenly distributed constraint that enforce each bit be evenly distributed on $+1$ and $-1$ over the dataset.
It's not too hard to show that this optimization is equivalent to minimizing $\Tr(Y^T(D-W)Y)$ where $Y \in \R^{n\times k}$ is the matrix containing codewords, $D$ is the degree matrix with $D_{ii} = \sum_{j=1}^n W_{ij}$. 
However, due to the binary constraint $Y(i,j) \in \{-1, 1 \}$, this probelm is NP hard ,so instead authors threshold a spectral
relaxation whose solution is the bottom $k$ eigenvectors of graph Laplacian matrix $L = D-W \in \R^{n \times n}$. 
This however provides a solution to only training datapoints. In order to extend it to out-of-samples, they assume datapoints are sampled from a separable probability distribution $p(x)$; using the fact that graph Laplacian eigenvectors converge to the Laplace-Beltrami eigenfunctions of manifolds, they set thresholded eigen functions as codewords. However they only examine the simple case of a \textit{multidimensional uniform distribution} or box shaped data, as these eigen functions are well-studied. 

In \cite{fergus2009semi}, Fergus \etal extended their previous work\cite{weiss2009spectral} to \textit{any} separable distribution, i.e. any distribution $p(x)$ with a product form.
They consider semi-supervised learning in a graph setting, where a labeled dataset of input-output pairs $(X_{m},Y_m) = \{(x_1,y_1),$
$\ldots, (x_m,y_m)\}$ is given, and they need to label a larger set $X_u = \{x_{m+1},\ldots,x_n\}$ of unlabelled points.
Authors form the graph of all datapoints $X_m \cup X_u$, where vertices represent datapoints and edges are weighted with a Gaussian function $W_{i,j} = \exp(- \|x_i - x_j\|/\sigma)$. 
The goal is to find functions $f$ which agree with labeled data but are also smooth with respect to the graph, therefore they formulate the problem as minimizing the error function $J(f) = f^TLf +  \lambda \sum_{i=1}^{\ell}(f(i)-y_i)^2$, where $f(i)$ is the embedding of $i$-th point and $\Lambda$ is a diagonal matrix whose diagonal elements are
$\Lambda_{i,i} = \lambda$ if $x_i$ is a labeled point and $\Lambda_{i,i} = 0$ otherwise.
Note that $f^TLf$ is the smoothness operator defined on the entire graph Laplacian as $f^TLf = 1/2 \sum_{i,j}W_{i,j}(f(i)-f(j))^2$, and $\lambda \sum_{i=1}^{\ell}(f(i)-y_i)^2$ represents the loss on the labeled data.
Similar to their previous work\cite{weiss2009spectral} authors approximate the eigen vectors of $L$ by eigen functions
of laplace-beltrami operator defined on probability distribution $p$.

Finally, in the most recent work of this series, ``Multidimensional Spectral Hashing''(MDSH) \cite{weiss2012multidimensional}, Weiss \etal introduced a new formulation for learning binary codes; unlike other methods that minimize Hamming distance $\|Y_{i,:} - Y_{j,:}\|$, MDSH  approximates original affinity $W_{i,j}$ with weighted Hamming affinity $Y_{i,:}^T \Lambda Y_{j,:}$, where $\Lambda = \diag(\Lambda_1,\cdots,\Lambda_k)$ gives a weight to each bit.
The authors show the best binary codes are obtainable via performing binary matrix factorization of affinity matrix, with the optimal weights given by the singular values. 

SSH and MDSH can be adapted to the streaming setting, but have the unsatisfactory elements that neither addresses approximating the initial matrix optimization directly. Moreover the hashing functions (eigenfunctions) they learn are wholly determined by the initial training set and do not adapt as more points are streamed in.

In another line of works, authors formulate the problem as an iterative optimization. 
In \cite{gong2011iterative}, Gong and Lazebnik suggest ``Iterative Quantization''(ITQ) algorithm which is an iterative approach based on alternate minimization scheme that first projects datapoints onto top $k$ right singular vectors of data matrix, and then takes the sign of projected vectors to produce binary codes. Authors show that if we consider projected datapoints as vectors in a $k$-dimensional binary hypercube $C \in \{-1,1\}^{k}$, then sign of vector entries in each dimension is determined by the closest vertex of hypercube along that dimension. As rotating this hypercube does not change the codes, they alternatively minimize the quantization loss $Q(B,R) = \|B-VR\|_F^2$ by fixing one of two variables $B$ the binary codes, or $R$ the rotation matrix, and solving for the other. They show that in practice repeating this for at most $50$ iterations beats some well-known methods including \cite{fergus2009semi,raginsky2009locality} and spectral hashing\cite{weiss2009spectral}.

Heo \etal ~cite{spherical} present an iterative scheme that partitions points using hyperspheres. Specifically, the algorithm places $k$ balls, such that the $i$th bit a point $S_{;i}$ is 1 if it is contained in the $k$-th ball and $0$ otherwise. At each step of the process if the intersection of any two balls contains too many points, a repulsive force is applied between them, whereas if the intersection contains too few an attractive force is applied. This minimization continues until a reasonably balanced number of the points are contained in each hypersphere.  Both these iterative algorithms seem difficult to adapt to a streaming setting, in the sense that the hash functions are expensive to learn on a training set and not easily updated.
%

In the supervised setting, Quadrianto \etal~\cite{quadrianto} present a probabilistic model for learning and extending binary hash codes. They assume the input dataset follows certain simple and well studied probability distributions, and that supervision is provided in terms of labels indicating which points are neighbors and which are far. Under these constraints, they may train a latent feature model to extend binary hash codes in the streaming setting as new data points are provided.

In the broader context, the most comparable line of works with our problem is matrix sketching in the stream.
Although there has been a flurry of results in this direction\cite{numerical,ghashami1,Lib12}, we mention those which are most related to our current work.
In ~\cite{drineas2005nystrom}, Drineas and Mahoney approximate a gram matrix $G \in \R^{n \times n}$ by sampling $s$ columns (datapoints) of an input matrix $A \in \R^{d \times n}$ proportional to the squared norm of columns. 
They approximate $G$ with $\tilde{G} = C W_k^{+} C^T$, where $C \in \R^{n \times s}$ is the gram matrix between $n$ datapoints and $s$ sampled points, $W_k \in \R^{s \times s}$ is the best rank $k$ to $W$ where $W$ is the gram matrix between sampled points. 
They need to sample $O(k/\eps^4)$ columns to achieve Frobenius error bound $\|G - \tilde{G}_k\|_{F} \leq \|G - G_k\|_{F} + \eps \sum_{i=1}^n G_{ii}^2$, and need to sample $O(k/\eps^2)$ columns to get spectral error bound $\|G - \tilde{G}_k\|_{2} \leq \|G - G_k\|_{2} + \eps \sum_{i=1}^n G_{ii}^2$. Their algorithm needs $O(d/\eps^2 + 1/\eps^4)$ space and has running time of $O(nd + n/\eps^2 + 1/\eps^4)$ on training set. The update time for any future datapoint (or test point) is $O(d/\eps^2 + 1/\eps^4)$.

The state-of-the-art matrix sketching technique is \FD algorithm first introduced by Liberty~\cite{Lib12} and then reanalyzed by Ghashami and Phillips~\cite{ghashami1}. \FD  maintains a deterministic, small space sketch for an input matrix and can be easily incrementally updated in the stream. In fact, for any input matrix $A \in \R^{n \times d}$, \FD maintains a sketch $B \in \R^{\ell \times d}$ with $\ell = 2/\eps$ rows, achieves error bound $\|A^TA - B^TB\|_2 \leq \eps \|A\|_F^2$ and runs in time $O(nd/\eps)$. It is shown by Woodruff that the approximation quality is optimal~\cite{lowerbound}. 

\subsection{Our Result}
 In this paper, we focus on finding codewords for a dataset $S \subset \mathbb{R}^d$ given in a stream. We consider an unsupervised setting where mutual similarity between datapoints is induced by Gaussian kernel function $w(p,q) = \exp(-\|p-q\|^2/\sigma)$ rather than any contextual information. 
We develop a reasonable model of data holding two assumptions:
\begin{enumerate}
\item \textit{sparsity},  that enforces data similarity not being dominated by ``near-duplicates''.
\item \textit{bounded doubling dimension}, that is data has a low-dimensional structure. This assumption is widely used  as ``effective low-dimension'' in Euclidean near neighbor search problems \cite{KRu,KLee,BKL06,CG06,IN07,HPK13-lowd,abdullah}, and corresponds well with existence of a good binary codebook \footnote{A good binary codebook is roughly equivalent to a low distortion embedding into a low-dimensional Hamming space.}.
\end{enumerate}
Under this model, we propose the ``Streaming Spectral Binary Coding'' (SSBC) algorithm that builds off of \fd and shows that if training set is a ``good representor'' of the stream, i.e. that the stream is in random order, then one can accurately update important directions (eigen vectors) of the weight matrix in a stream. These vectors are then used to construct the desired codewords.

In fact, as we show in section \ref{sec:exp} our technique works in both streaming and online settings, achieves $O(n/\eps \polylog n)$ space in former setting and $O( k/\eps \polylog n)$ space in latter setting. Note both bounds are much smaller than $\Theta(n^2)$ which is the required space for storing similarity matrix. 
%
%

Our starting matrix optimization formulation is closely adapted from those posed in this line of work by Fergus, Weiss and Torralba. 
However, our approach to \emph{solving} the problem and out-of-sample extension differs fundamentally from previous tactics of using functional approximation methods and learned eigenfunctions. We maintain a sketch of the weight matrix instead, and adjust it during the course of the stream. 
While known functional analysis techniques rely on assumptions on the data distribution (in particular that it is drawn from a separable distribution) we argue that solving for the matrix approximation directly addresses the original optimization problem without such restrictions, thereby achieving the superior accuracy our experiments demonstrate.

\section{Setup and Algorithm}\label{sec:algorithm}
In this section, we first set up matrix optimization problem that is the starting point of the work by Weiss, Fergus and Torralba~\cite{weiss2012multidimensional}, then we describe our algorithm ``Streaming Spectral Binary Coding'' (SSBC) for approximating binary codewords in a stream or online setting.
\subsection{Model and Setup}\label{sec:model}  
We denote input dataset as $S \in \R^{n \times d}$ containing $n$ datapoints in $\R^d$ space and represent binary codes as $Y \in \R^{n \times k}$, where $k \in \Z^+$ is a parameter specifying length of codewords.

We define affinity or similarity between datapoints $S_{i,:}$ and $S_{j,:}$ as $w(S_{i,:},S_{j,:}) = \exp(-\|S_{i,:} - S_{j,:}\|^2/\sigma^2)$ where $\sigma$ is a parameter set by user,  corresponding to a threshold between ``near" and ``far" distances.
Since codewords are vectors with $\pm1$ entries, one can write $\|Y_{i,:} - Y_{j,:}\|^2 = 2k - 2Y_{i,:}^TY_{j,:}$, and match \textit{Hamming affinity} $Y_{i,:}^TY_{j,:}$ with $w(S_{i,:},S_{j,:})$  instead of minimizing Hamming distance. 
Similar to~\cite{weiss2012multidimensional}, we define a diagonal weight matrix $\Lambda = [\Lambda_{1,1},\cdots,\Lambda_{k,k}]$ to give an importance weight $\Lambda_{j,j}$ to $j$-th bit of codewords.
Therefore we formulate the problem as:
\begin{align*}
(Y^*, \Lambda^*) &= \argmin_{Y_{i,:} \in \{\pm1\}^k,\Lambda} \sum_{i,j} \left(w(i,j) - Y_{i,:}^T \Lambda Y_{j,:}\right)^2 \\
&= \argmin_{Y_{i,:} \in \{\pm1\}^k,\Lambda} \|W - Y\Lambda Y^T\|_F^2
\end{align*}

This optimization problem is solvable by a binary matrix factorization of the affinity matrix, $W$. As discussed in \cite{srebro2003weighted,weiss2012multidimensional}, the $\pm1$ binary constraint makes this problem computationally intractable, but a relaxation to real numbers results in a standard matrix factorization problem that is easily solvable. 
If $W = U \Lambda U^T$ is eigen decomposition of $W$, then $i$-th row of $U_{k} \in \R^{n \times k}$ provides a codeword of length $k$ for $i$-th datapoint, which can be easily translated into a binary codeword by taking sign of entries.
The result binary codeword will be an approximation to the solution of binary matrix factorization. 

We consider solving binary encoding problem in two settings ``streaming'' and ``online'', where in both model one datapoint arrives at a time, is processed quickly and not read again. In the streaming setting, we output all binary codewords at the end of stream, while in the online setting, we are obliged to output binary codeword of current datapoint before seeing next datapoint.
Space usage is highly constrained in both models, so we cannot store the entire weight matrix $W$ (of size $\Omega(n^2)$) nor even the dataset itself (of size $O(nd)$).   

Below, we specify assumptions we make in our data model for the purposes of theoretical analysis. However, we note that our experiments show strong results without enforcing any restrictions on the datasets we consider.

\begin{enumerate}
\item{Our first assumption is ``sparsity", namely that no two points $p$ and $q$ are asymptotically close to each other. Specifically, that $\|p - q \| \geq (0.1\; \sigma/\log n)$, for all $1 \leq i,j \leq n$, where $\sigma$ is the threshold distance parameter of our Gaussian kernel. When our data is being analyzed for clustering/near neighbor purposes, this condition implies that identical points have either been removed or combined into a single representative point.}
\item{Our second assumption is that the data has bounded doubling dimension $d_0$. Namely that a ball $B$ of radius $r$ contains at most $O \left( \left( r/\eps \right)^{d_0}\right)$ points spaced at distance at least $\eps$. This is a standard model in the algorithms community for modeling data drawn from a low dimensional manifold. It is also intuitively compatible with the existence of a good representation of our data by $k$-bit codewords for bounded $k$, as binary encoding is simply an embedding into $k$-dimensional Hamming space.}
\end{enumerate}

%
%

\subsection{Streaming Binary Coding Algorithm}
Our method, which we refer to as ``SSBC'' is described in algorithm \ref{alg:ssh}.
SSBC takes three input values $S_{train}$, $S_{test}$ and $k$ where $S_{train}$ is a small training set sampled uniformly at random from the underlying distribution of data, e.g. $\mu$. We denote size of $S_{train}$ by $|S_{train}| = m$, and we assume $m > \polylog(n)$. For ease of analysis, wherever we come across some $\polylog(n)$ to a constant exponent, we assume the term to be smaller than $m$.
On the other hand, $S_{test}$ is a potentially unbounded set of data points coming from same distribution $\mu$. Even though $S_{test}$ can be unbounded, for the sake of analysis, we denote total number of datapoints in union of both sets as $n = |S_{train}| + |S_{test}|$.
Value $k > 0$ is the length of the codewords we seek.

The algorithm maintains a small sketch $B$ with only $\ell \ll m \ll n$ rows. For each datapoint $p \in S_{train}$, SSBC computes its (Gaussian) affinity with all points in $S_{train}$, outputs an $m$ dimensional vector $\hat w_p$ as the result, and inserts it into $B$. Once $B$ is full, SSBC takes \svd of $B$ ($[U,\Sigma,V] = B$), subtracts off smallest singular value squared, i.e. $\Sigma_{\ell, \ell}^2$, from squared of all singular values, and reconstruct $B$ as $B = \Sigma' V^T$. This results in zeroing out last row of $B$, and making space for processing next upcoming train datapoint. Note after processing $S_{train}$, matrix $V \in \R^{\ell \times \ell}$ contains an $\ell$-dimensional approximation to similarity structure of train set. 
As we observe, SSBC employs \FD algorithm \cite{Lib12} to process affinity vectors in streaming manner; instead of referring to \FD, we included its pseudocode completely in algorithm \ref{alg:ssh}. 
As many similarity measures can be used to capture the affinity between datapoints, SSBC uses \textit{Gaussian affinity} $W=\exp \left(-\|q-p\|^2/\sigma \right)$, where $\sigma$ is a parameter denoting the average near neighbor distance we care about. 
This function is called in subroutine \ref{alg:ha} to measure the affinity between any test point and all train datapoints.

At any point in time, we can get binary codeword of any datapoint $q \in S$ by first computing its affinity with $S_{train}$, getting vector $\hat{w}_q$ as output and multiplying it by right singular vectors. More specifically if $y_q$ denotes binary codeword of $q$, then $y_q = sign(\hat{w}_q \times V) \in \R^{\ell}$ gives a $\ell$-length codeword. To get a codeword of length $k$, we truncate $V$ to its first $k$ columns, $V_k \in R^{\ell \times k}$. 


\begin{algorithm}
\caption{\label{alg:ssh} Streaming Spectral Binary Coding (SSBC)}
\begin{algorithmic}
\STATE \textbf{Input:} $S_{train}, S_{test} \subset \R^{d}$, $k \in \Z^+$ as length of codeword
\STATE Define $S = [S_{train};S_{test}]$, $m = |S_{train}|$, and $n = |S|$
\STATE Set $\ell = \lceil k + k/\eps \rceil$ as sketch size
\STATE Set $B \in \R^{\ell \times m}$ to full zero matrix
\FOR {$i \in [1:n]$}
  \STATE $\hat w = \HA(S_{i,:}\; , S_{train}$)
  \STATE Insert $\hat w$ into a full zero row of $B$
  \IF {$B$ has no full zero rows}
    \STATE $[U,\Sigma,V] = \svd(B)$
    \STATE $\Sigma' = \sqrt{\Sigma^2 - \Sigma_{\ell,\ell}^2}$
    \STATE $B = \Sigma'V^T$
  \ENDIF
\ENDFOR
\STATE \textbf{Return} $B$
\end{algorithmic}
\end{algorithm}

A notable point about SSBC is that it can construct binary codewords on the fly in an online manner, i.e. using current iteration's matrix $V$ to generate the binary codeword for current datapoint. 
As we show in section \ref{sec:err} this leads to the small space usage of $O(\ell m) = O(1/\eps^2 \polylog(n))$. Clearly, SSBC can generate all codewords at the end of stream too (streaming setting); in that case it needs to store all $\hat w$ vectors and uses final matrix $V$ to construct codewords. Space usage in streaming setting is $O(n m + \ell m) = O\left((1/\eps^2 + n)\polylog(n)\right)$.
The update time (or test time) in both models is $O(md + d\ell) = O(d/\eps \polylog(n))$. 

\begin{algorithm}
\caption{\label{alg:ha} Gaussian Affinity}
\begin{algorithmic}
\STATE \textbf{Input:} $q \in \R^{d}$ as a test point, $S_{train} \subset \R^d$
\STATE Define $\sigma$ to similarity threshold between points in $S_{train}$
\STATE Set $\hat W \in \R^m$ to zero vector, where $m = |S_{train}|$
\STATE Set $i = 0$
\FOR {$p $ in $S_{train}$}
  \STATE $\hat W[i] = \exp(-\|q - p\|^2/\sigma)$
  \STATE $i\; \Scale[0.8]{++}$
\ENDFOR
\STATE \textbf{Return} $\hat{W}$
\end{algorithmic}
\end{algorithm}

To explain good performance of SSBC, we argue that under the data model described in Section~\ref{sec:model}, squared norms of the columns of $W$ are within a $\polylog(n)$ factor of each other. Using this fact, we show a uniform sample of the columns of $W$ is a good approximation to $W$. In what follows, let $C_i$, $C_{max}$ and $C_{min}$ denote squared norm of $i$-th column of $W$, maximum and minimum squared norm of any column of $W$ respectively.

\begin{lemma}
Under ``sparsity'' and ``bounded doubling dimension'' assumptions:
\[
C_{max}/C_{min} \leq (\log(n))^{O(d_0)}
\]
\end{lemma}
\begin{proof}
First note that it is trivially true that $C_{min} \geq 1$, since $C_{min}^2 \geq W_{i,i}^2 =  \exp^2(-\|S_{i,:} - S_{i,:} \|^2) = \exp^2(0) = 1$. We now upper bound $C_{max}$. Let $C_i$ denote squared norm of an arbitrary column of $W$, so that upper bounding $C_i$ would also bound $C_{max}$. Let $S_{i,:}$ be the corresponding datapoint associated with column $W_{:,i}$.
We proceed by partitioning points of $S$ close to (similar) and far (dissimilar) from $S_{i,:}$ as $P_c$ and $P_f$, respectively.
Define $P_f = \{ S_{j,:} \in S$, s.t. $W_{i,j} \leq \frac{1}{n} \}$ and $P_c = \{ S_{j,:} \in S$, s.t. $W_{i,j} \geq \frac{1}{n} \}$. Note that the contribution of $P_f$ to $C_i$ is at most $| P_f | \frac{1}{n} \leq 1$, and contribution of $P_c$ to $C_i$ is at most 
 $| P_c| \cdot 1 \leq |P_c|$. So we bound the size of $P_c$. First we upper bound distance of any point $S_{j,:} \in P_c$ to point $S_{i,:}$ as following:
\[
W_{i,j} = \exp \left(-\frac{\|S_{i,:} - S_{j,:} \|^2}{\sigma} \right) \geq \frac{1}{n}
\]
Therefore $\|S_{i,:} - S_{j,:} \|^2  \leq \sigma \ln n$.

Now considering the sparsity condition, we have that the number of points $S_{j,:}$ within $\sigma \ln n$ distance of $S_{i,:}$ is at most 
$\left(\frac{\sigma \ln n} {0.1 \sigma} \right)^{d_0} \leq (\log n)^{O(d_0)}$.
\end{proof}
We immediately get the following corollary as a consequence:
\begin{corollary}\label{cor:first}
It holds $\forall i$, $1 \leq i \leq n$ that $$\frac{1}{\polylog(n)} \frac{ \| W \|_F^2 }{n} \leq C_i \leq \polylog(n) \frac{ \| W \|_F^2 }{n}$$.
\end{corollary}
\begin{proof}
For the upper bound, we have $n C_{min} \leq \| W \|_F^2$, or $C_{min} \leq \frac{\| W \|^2}{n}$. But for  arbitrary $C_i$, we have $C_i \leq polylog(n) C_{min}$ and hence $C_i \leq polylog(n) \frac{ \| W \|_F^2 }{n}$. The lower bound on $C_i$ follows similarly using $C_{max}$.
\end{proof}

\section{Error Analysis}\label{sec:err}
In this section, we prove our main result. 
Let $W \in \R^{n \times n}$ be the exact affinity matrix of $n$ datapoints in $S$, where $W_{i,j} = \exp(-\|S_{i,:} - S_{j,:}\|^2/\sigma)$. 
Let $m = |S_{train}|$ be size of training set and $\hat W \in \R^{n \times m}$ be the rescaled affinity matrix between all points in $S$ and $S_{train}$. 
Under the assumption that $S_{train}$ is drawn at random, we can imagine $\hat W$ is a column sample drawn uniformly at random from $W$.
In the general case, column samples are only good matrix approximations to $W$ if each column is drawn proportional to its norm, which is not known in advance in streaming setting. However we show that under our data model assumptions of Section \ref{sec:model}, a uniform sample suffices.
%
Define $W_{:, j^*}$ to be the column of $W$ that gets sampled for $j$-th column of $\hat W$. (This corresponds to a choice of $S_{j^*,:}$ as the $j$-th point in $S_{train}$). Now define the scaling factor of $\hat{W}$ as $\hat{W}_{i,j} = \frac{\sqrt{n}}{m} W_{i, j^*}$. 
Define $\tilde W = \hat W B^T B \hat{W}^\dagger$ as approximated affinity that could be constructed from $\hat{W}$ and the output of SSBC, i.e. $B \in \R^{\ell \times m}$. 
\footnote{Our algorithm does not actually construct $\hat W$ and $\tilde W$. Rather we use them as existential objects for our theoretical analysis.}

 We show that for $m = \Omega(\frac{1}{\eps} \polylog(n) \log(1 / \delta))$ and $\ell = 2/\eps$, then $\|W^2 - \tilde W\|_2 \leq \eps \|W\|_F^2$ with probability at least $1-\delta$.
In our proof we use the Bernstein inequality on sum of zero-mean random matrices, which is stated below.

\paragraph{Matrix Bernstein Inequality}
Let $E_1,\cdots,E_m \in \R^{n\times n}$ be independent random matrices such that for all $1 \leq i \leq m$, $\E[E_i] = 0$ and $\|E_i\|_2 \leq \Delta$ for a fixed constant $\Delta$. If we define variance parameter as \[
\sigma^2 ~:= \max \{\|\sum_{i=1}^m \E[E_i^T E_i]\|_2,\|\sum_{i=1}^m \E[E_i E_i^T]\|_2 \}
\]
Then for all $t \geq 0$:
\[
\Pr\left[\Big\|\sum_{i=1}^m E_i\Big\|_2 \geq t \right] \leq 2n \cdot \exp\left(\frac{-t^2}{3\sigma^2 + 2\Delta t}\right)
\]
Lemma below bounds spectral error between $W$ and $\hat W$.
\begin{lemma}
\label{lem:cs_bernstein}
If $W \in \R^{n \times n}$ is the exact affinity matrix of points $S$ and $\hat{W} \in \R^{n \times m}$ is the affinity matrix between points in $S$ and $S_{train}$, then for $m = \Omega \left(  \frac{1}{\eps} \polylog(n) \log(1/\delta) \right) $

\[
\|W^2-\hat{W}\hat{W}^T\|_2 \leq \eps \|W\|_F^2
\]
holds with probability at least $1-\delta$.
\end{lemma}
\begin{proof}
Consider $m$ independent random variables $E_i = \frac{1}{m} W^2 - \hat{W}_{:,i}\hat{W}_{:,i}^T$. We can show $\E[E_i] = 0$ as follows
\begin{align*}
\E[E_i] &= \frac{1}{m} W^2 - \E[\hat{W}_{:,i}\hat{W}_{:,i}^T]\\
&= \frac{1}{m} W^2 - \sum_{j=1}^n \frac{m}{n} \left(\frac{\sqrt{n}}{m} \right)^2 W_{:,j^*} W_{:,j^*}^T\\
&= \frac{1}{m} W^2 - \frac{1}{m} \sum_{j=1}^n W_{:,j}W_{:,j}^T \\
&= \frac{1}{m} W^2 - \frac{1}{m} WW^T = 0
\end{align*}
Note that last equality is correct because $W$ is a symmetric matrix, and therefore $W^2 = WW^T$.
We can now bound $\E[\hat W \hat{W}^T] = \sum_{i=1}^m \E[\hat{W}_{:,i}\hat{W}_{:,i}^T] = \sum_{i=1}^m \frac{1}{m} W^2 = W^2$. Using this result we bound $\|E_i\|_2$ as follows
\begin{align*}
\|E_i\|_2 &= \left\|\frac{1}{m} W^2 - \hat{W}_{:,i}\hat{W}_{:,i}^T \right\|_2 \\
&=
\left\|\frac{1}{m} \E[\hat{W}\hat{W}^T] - \hat{W}_{:,i}\hat{W}_{:,i}^T\right\|_2 \\
&\leq 
\frac{1}{m} \left\|\E[\hat{W}\hat{W}^T]\right\|_2 + \|\hat{W}_{:,i}\hat{W}_{:,i}^T\|_2 \\
&\leq 
\frac{1}{m} \E\left[\|\hat{W}\|_2^2\right] + \|\hat{W}_{:,i}\|_F^2  \\
&\leq 
\frac{1}{m} \E\left[\|\hat{W}\|_F^2\right] + \frac{n}{m^2} \|W_{:,i^*}\|_F^2 \\
&\leq 
\frac{1}{m} \|W\|_F^2 + \frac{n}{m^2} \left(\polylog(n)\frac{\|W\|_F^2}{n}  \right) \\
&= \frac{1}{m}\|W\|_F^2 + \frac{\polylog(n)}{m^2}\|W\|_F^2 \\
&= O \left( \left(\frac{\polylog(n)}{m^2} + \frac{1}{m}  \right)\|W\|_F^2 \right)
\end{align*}
Where the fourth line is achieved using Jensen's inequality on expected values, which states $\|\E[X]\| \leq \E[\|X\|]$ for any random variable $X$ and the third last line by Corollary \ref{cor:first}. 
Therefore $\Delta = \|W\|_F^2/m + \|W\|_F^2 \polylog(n)/m$ for all $E_i$s.

In order to bound variance parameter $\sigma^2$, first note due to symmetry of matrices $E_i$, its definition reduces to 
\[
\sigma^2 = \left\|\E\left[\sum_{i=1}^m E_i^2\right]\right\|_2 \leq \E \left[ \left\| \sum_{i=1}^m E_i^2 \right\|_2\right] \leq m\; \E \left[ \left\|E_i^2 \right\|_2\right]
\]
Where the last step follows since all the $E_i$ are identical random variables. We already have an upper bound on the value $\|E_i \|_2$ may achieve, and hence the square of this upper bounds $\E [\|E_i^2 \|_2]$. We bound $m\; \E[ \|E_i^2\|_2]$ as follows:
\begin{align*}
m\; \E \left[ \| E_i^2 \|_2  \right] &\leq m\; \|E_i \|_2^2  \\ 
& \leq  m\; O \left(\ \left(\frac{\polylog(n)}{m^2} + \frac{1}{m} \right) \|W\|_F^2 \right)^2\\
& \leq  O \left(\frac{\polylog(n)}{m}\|W\|_F^4 \right)
\end{align*}

Setting $M = \sum_{i=1}^m E_i = W^2 - \hat{W}\hat{W}$ and using Bernstein inequality with $t = \eps \|W\|_F^2$ we obtain 
\begin{align*}
&\Pr \left[\|W^2-\hat W \hat{W}^T\|_2 \geq \eps \|W\|_F^2 \right] \\
&\leq 2n\exp \Scale[1.2]{\left(\frac{-(\eps \|W\|_F^2)^2}{3\|W\|_F^4  \polylog(n)/m^2 + 2\eps \|W\|_F^4 \left(\frac{1}{m} + \frac{\polylog(n)}{m^{2}} \right)} \right)}\\
& = 2n \exp \left( \frac{-\eps^2 m^2}{3 \polylog(n) + 2\eps (m + \polylog(n))} \right) \leq \delta
\end{align*}
Taking natural logarithm from both sides and inverse ratios, we get:
\[
\frac{3\polylog(n) + 2\eps \polylog(n)}{\eps^2 m^2} + \frac{2}{\eps m} \leq \ln^{-1}{(2n/\delta)}
\]
Considering that $\eps \leq 1$, we seek to bound:
\[
\frac{3 \polylog(n)}{\eps^2 m^2} + \frac{2}{\eps m} \leq \ln^{-1}{(2n/\delta)}
\]
Solving for $m$ we obtain that for $m = \Omega \left(\frac{1}{\eps}\polylog (n) \log(1/\delta) \right)$, the bound holds with probability at least $1 - \delta$.
\end{proof}

Hence $W$ has a similar spectrum to $\hat{W}$. 
In the lemma below, we argue that spectrum of $\hat{W}$ can be captured well by sketch $B$. To this end we define $\tilde W = \hat W B^TB \hat{W}^{\dagger}$ and show $\tilde{W}$ is again similar to $\hat{W}$ 
; intuitively $\tilde{W}$ approximates projection of $\hat{W}$ onto the right singular vectors of $B$.
\begin{lemma}
\label{lem:FD}
Let $\hat W$ be the affinity matrix between datapoints in $S$ and $S_{train}$. Then for $\tilde W = \hat W B^TB \hat{W}^{\dagger}$ and $\ell = O \left( \frac{1}{\eps} \right)$
\[
\|\hat{W}\hat{W}^T - \tilde{W}\|_2 \leq \eps \|W\|_F^2
\]
\end{lemma}
\begin{proof}
We can bound $\|\tilde{W} - \hat{W}\hat{W}^T\|_2$ as following:
\begin{align*}
\|\tilde{W} - \hat{W}\hat{W}^T\|_2 &= \|\hat{W}B^TB\hat{W}^{\dagger} - \hat{W}\hat{W}^T \|_2  \\
&= \| \hat{W}(B^TB - \hat{W}^T\hat{W})\hat{W}^{\dagger} \|_2 \\
&\leq \|\hat{W}\|_2 \|B^TB - \hat{W}^T\hat{W}\|_2 \|\hat{W}^{\dagger}\|_2 \\
&=\|B^TB - \hat{W}^T\hat{W}\|_2 \\
&\leq \|\hat{W}\|_F^2/\ell
\end{align*}
And we can also bound $\|\hat{W}\|_F^2$:
\begin{align*}
\|\hat W\|_F^2 &= \sum_{i=1}^m \|\hat{W}_{:,i}\|^2 = \sum_{i=1}^m \frac{n}{m^2} \|W_{:,i^*}\|^2\\
&\leq \sum_{i=1}^m \frac{n}{m^2} \polylog(n) \frac{\|W\|_F^2}{n} \\
&= \sum_{i=1}^m \polylog(n) \frac{\|W\|_F^2}{m^2} = \frac{\polylog(n)}{m} \|W\|_F^2
\end{align*}
Putting the two bounds together we get:
\[
\|\tilde{W} - \hat{W}\hat{W}^T\|_2 \leq \frac{\polylog(n)}{m \ell} \|W\|_F^2
\]
Since we already showed $m > \polylog(n)$ in lemma \ref{lem:cs_bernstein}, setting $\ell = \frac{1}{\eps}$ suffices to complete the proof.
%
%
\end{proof}

\begin{theorem}\label{thm:eq}
Let $W \in \R^{n \times n}$ be similarity matrix of $S \in \R^{n \times d}$, and $\tilde W = \hat W B^TB \hat{W}^{\dagger}$ be the weight matrix constructed by $\hat W \in \R^{n \times m}$ and $B \in \R^{\ell \times m}$, where $\hat W$ is the set of columns sampled with replacement from $W$, and $B$ is the output of algorithm \ref{alg:ssh}. Then for $m = \Omega \left(  \frac{1}{\eps} \polylog(n) \log(1/\delta) \right) $ and $\ell = O \left(\frac{1}{\eps} \right)$:
\[
\|W^2 - \tilde W\|_2 \leq \eps \|W\|_F^2
\]
holds with probability $1-\delta$.
\end{theorem}
\begin{proof}
Having results of Lemmas \ref{lem:cs_bernstein} and \ref{lem:FD}, assuming $m$ to be sufficiently large to meet conditions of both Lemmas and using triangle inequality and rescaling $\eps$ to $\eps/2$ proves the result.
\end{proof}

We explain some informal intuition of what Theorem \ref{thm:eq} implies. First we could infer $\|W - \sqrt{\tilde W} \|_2$ 
is small.  Now writing the SVD decomposition of 
$B^T B$ as $U S U^T$, we get $\tilde W = \hat{W} U S U^T \hat{W}^{\dagger}  = \hat{W} U S U^T \hat{W}^{\dagger}$. The intuition then is that if $\hat{W}^{\dagger}$ is similar to $\hat{W}^T$, then $\sqrt{\tilde W} \approx \hat W U S^{1/2}$, which is just the projection of $\hat W$ onto the right singular space of $B$. 
This suggests that the right singular space of $B$ captures most of the spectrum of $W$, in the sense that a column sample of $W$ projected on the right singular space of $B$ and scaled appropriately recovers $W$ closely. 

\section{Experiments}
\label{sec:exp}
Herein we describe an extensive set of experiments on a wide variety of large input data sets.
We ran all algorithms under a common implementation framework using Matlab to have a fair basis for comparision.

We compared efficiency and accuracy of our algorithm (SSBC) versus well-known streaming binary encoding techniques, including ``Multidimensional Spectral Hashing''(MDSH)\cite{weiss2012multidimensional},``Locality Sensitive Hashing'' (LSH)\cite{datar2004locality} and ``Spectral Hashing'' (SH)\cite{weiss2009spectral}.

We also compare accuracy of these algorithms against exact solution for binary coding problem when posed as a matrix optimization. As the exact solution, we compute the affinity matrix of whole dataset $S_{total} = [S_{train} ; S_{test}]$, and take the eigen decomposition of that. Let $W_{total} \in \R^{n\times n}$ denote the affinity matrix for $S_{total}$. If $W_{total} = U \Lambda U^T$ is the eigen decomposition of $W_{total}$, then $i$-row of $sign(U_k)$ matrix provides a binary code of length $k$ to $i$-th datapoint in $S_{total}$. In our experiments, we considered two types of thresholding on exact solution, namely deterministic rounding and randomized rounding. The deterministic rounding version is called ``Exact-D'' in the plots, and it basically takes the sign of $U_k$ only. The randomized rounding one is called ``Exact-R'' in the plots, and what it does is that after computing $U_k$ it multiplies it by a random rotation matrix $R \in \R^{n \times n}$, and then takes the sign of entries. 

\paragraph{Datasets}
We compare performance of our algorithm on both synthetic and real datasets.
Each data set is divided into two subsets, $S_{train}$ and $S_{test}$, with same number of dimensions and different number of datapoints. 
Table \ref{tbl:datasets} lists all datasets along with some statistics about them.
We refer to each set as an $n \times d$ matrix $A$, with $n$ datapoints and $d$ dimensions. Training Set is taken small in size so that it easily fits into memory, while $S_{test}$ is a large stream of data whose datapoints are processed one-by-one by our algorithm.    

As synthetic dataset we used multidimensional uniform distribution with $d = 50$ dimensions in which $t$-th dimension $\forall t,\; 1 \leq t \leq d$ has a uniform distribution in range $[0,(1/t)^2]$. In spectral hashing algorithm\cite{weiss2009spectral}, authors argue their learned eigenfunctions converge most sharply for rectangle distribution and include experimental results on uniform distributions demonstrating this efficacy. We added such dataset here so as to evaluate SSBC for a dataset model well suited to their algorithm.

\begin{table}[t!!!!]
\begin{center}
\begin{tabular}{|c||c|c|c|c|c|}
\hline
\textbf{DataSet} & \textbf{\# Train} & \textbf{\# Test} & \textbf{Dimension} & \textbf{Rank} \\
\hline
\hline
\textsf{PAMAP} & 100 & 21000 & 44 & 44 \\
\hline
\textsf{CBM}& 200 & 11000 & 18 & 16  \\  
\hline   
\textsf{Uniform} & 500 & 10000 & 50 & 50 \\
\hline 
\textsf{Covtype} & 500 & 20000 & 54 & 53 \\  
\hline 
\end{tabular} 
\end{center}
\vspace{-2mm}
\caption{\label{tbl:datasets} Datasets Statistics.}
\end{table}

We used three real-world datasets in our experiments. In each dataset, we uniformly sampled a small subset of data at random and considered it as $S_{train}$, and used a subset of remaining part as $S_{test}$. Information about size of training set and test set is provided in table \ref{tbl:datasets}.
First real-world dataset was the famous \s{Covtype}\cite{Covtype} that contains information about predicting forest cover type from cartographic variables.
Second one was \s{CBM} or ``Condition Based Maintenance of Naval Propulsion Plants''\cite{Coraddu2013Machine} which is a dataset generated from simulator of a gas turbine propulsion plant. It contains $11934$ datapoints in $d = 16$ dimensional space.

The \s{PAMAP}\cite{reiss2012introducing} dataset is a Physical Activity Monitoring dataset that contains data of $18$ different physical activities (such as walking, cycling, playing soccer, etc.), performed by $9$ subjects wearing $3$ inertial measurement units and a heart rate monitor. The dataset contains $54$ columns including a timestamp, an activity label (the ground truth) and $52$ attributes of raw sensory data. In our experiments, we removed columns containing missing values and used a subset with $d = 44$ columns.

\paragraph{Metrics}
We use three following metrics to compare accuracy of discussed algorithms:
\begin{itemize}
\item \textit{Precision}: The number of true similar datapoints returned by an algorithm over total number of datapoints returned by the algorithm.
\item \textit{Recall}: The number of true similar datapoints returned by an algorithm over correct number of similar datapoints.
\item \textit{Mean Average Precision (MAP)}: The mean of the average precision scores for each test point.
\end{itemize}

We have used the Guassian function $w(p,q) = \exp(-\|p-q\|^2/\sigma)$ to compute affinity between any two datapoints $p$ and $q$. We set $\sigma$ in each dataset to the average distance of all train datapoints to their $30$-th nearest neighbour, and set this threshold in both Hamming and Euclidean space to designate whether two points are similar. We refer to this parameter as $\sigma_{30}$. We have used $\sigma_{30}$ in all the experiments involving ``precision" and ``recall" metrics. 
For Mean Average Precision(MAP) metric, we consider $3$ different similarity levels comprising $\sigma_{30}$, the average of all pairs distance in training set ($\sigma_{all}$) , and $\sigma_{30} / 4$. In all cases, we set the choice of the $\sigma$ parameter in our Gaussian weight kernel equal to our similarity threshold for classifying points as near. The number of bits we use ranges from $k = 20$ to $k=50$ with increments of $5$.


\begin{figure}[t!]
\begin{centering}
\includegraphics[width=\figsize]{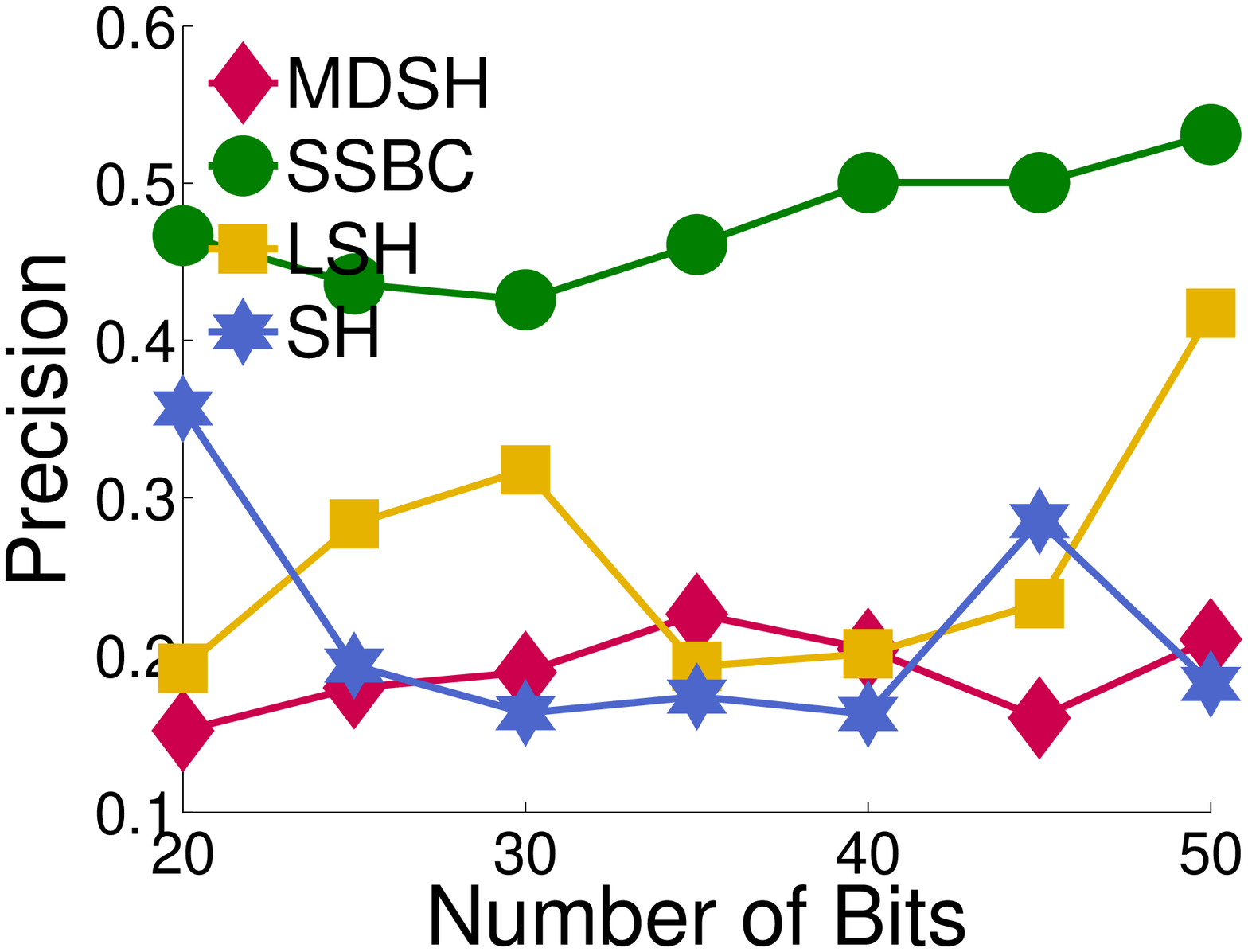}
\includegraphics[width=\figsize]{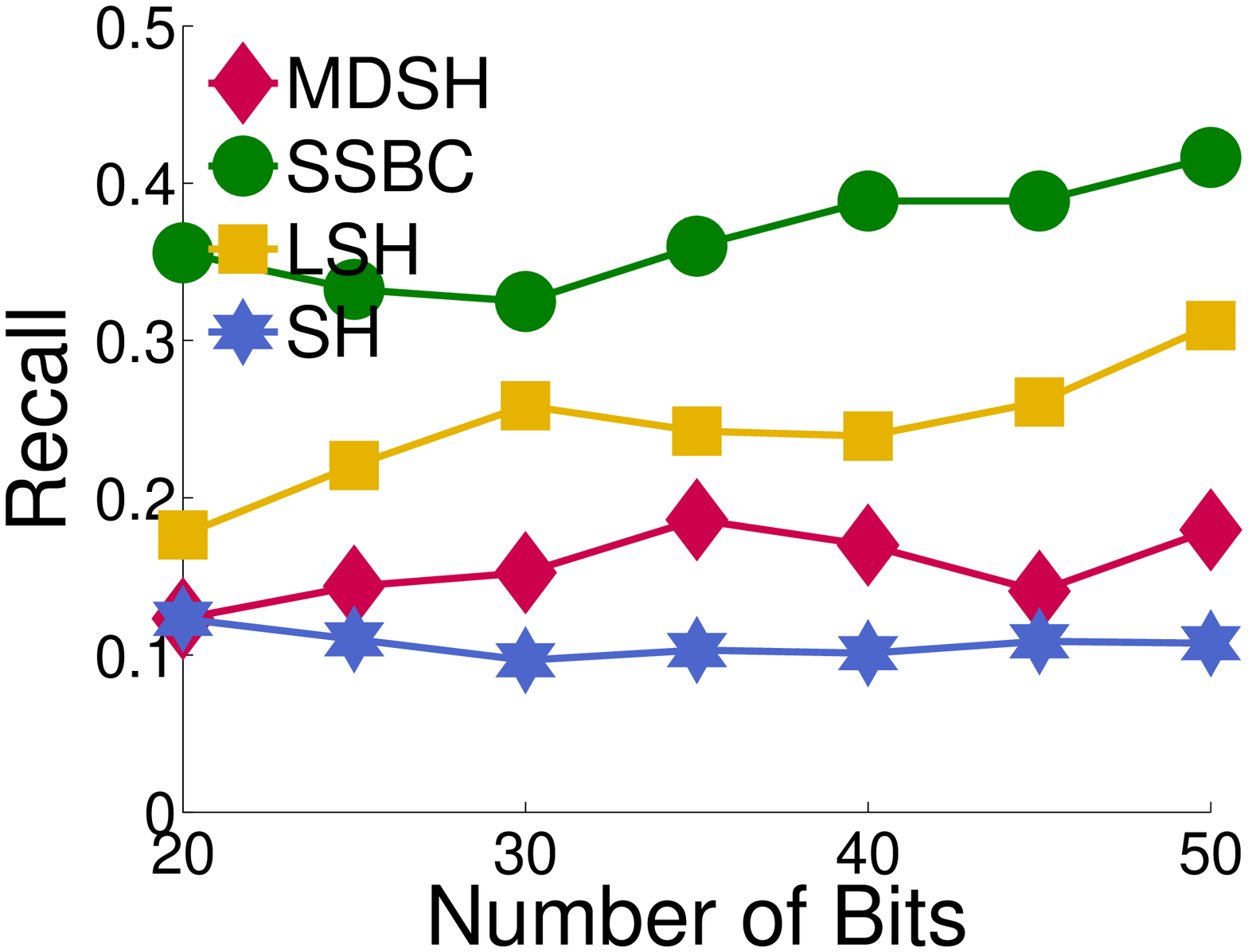}
\includegraphics[width=\figsize]{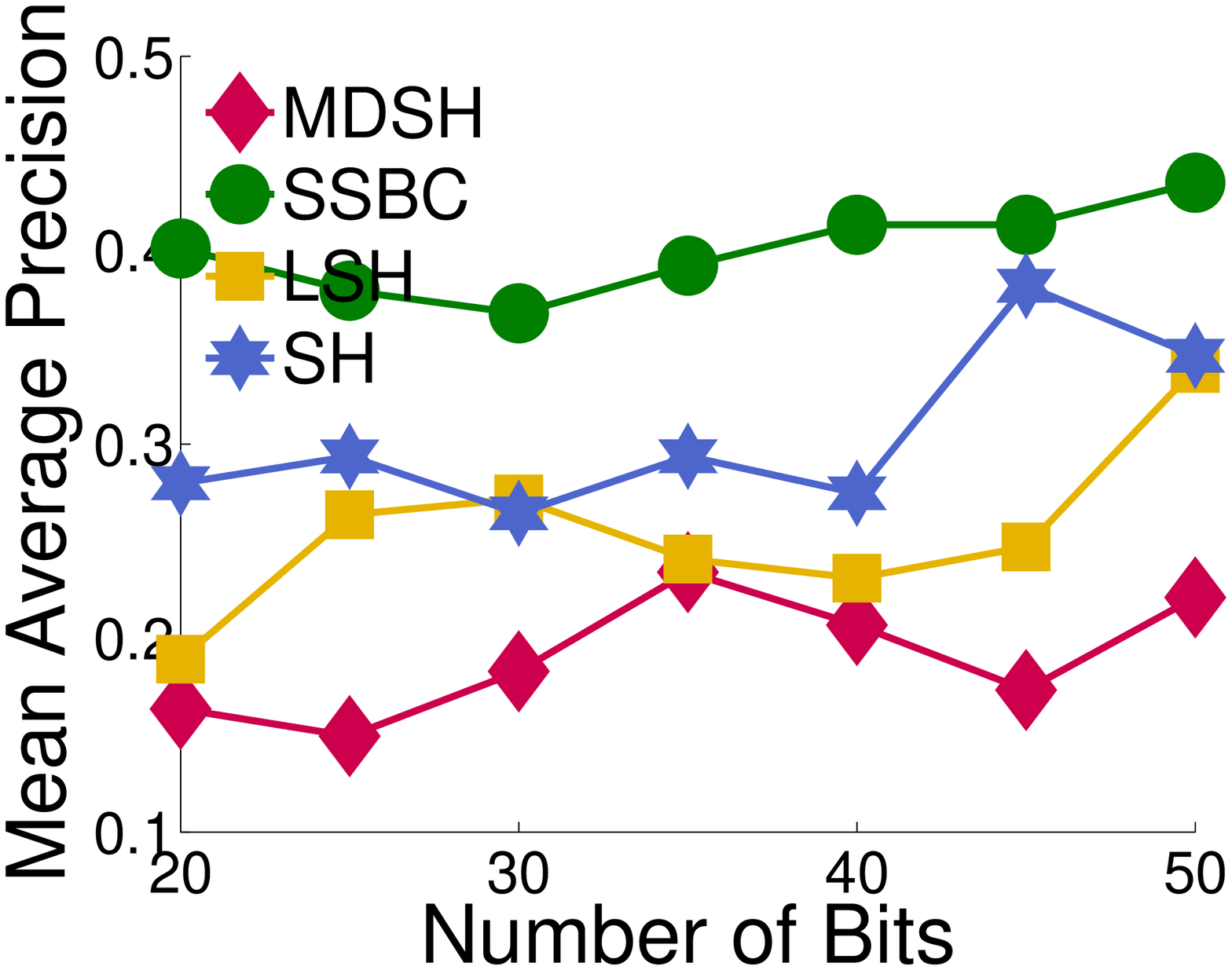}
\includegraphics[width=\figsize]{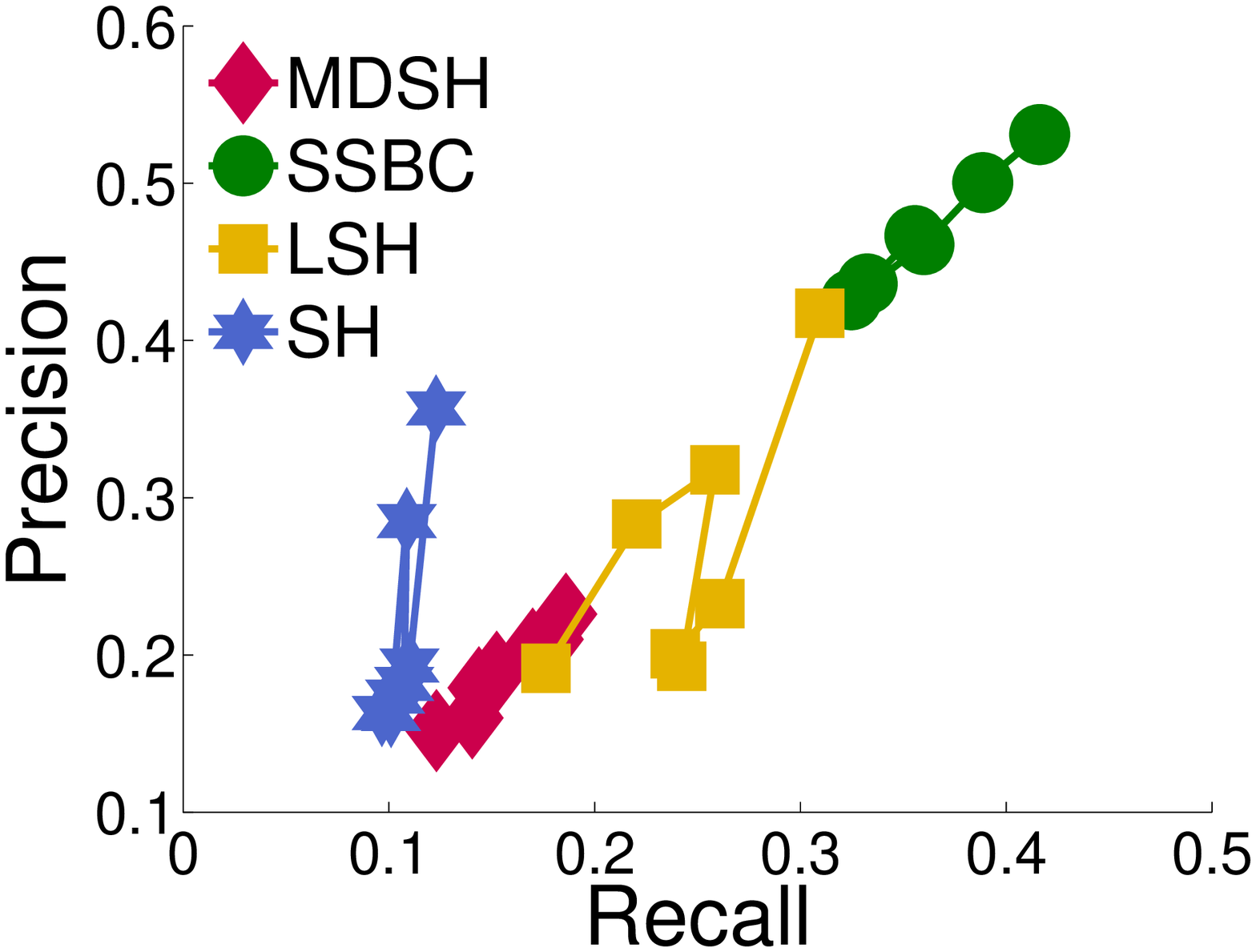}
\vspace{-3mm}
\caption{\label{fig:pamap}
Results on \s{PAMAP} dataset.}  
\vspace{-2mm}
\end{centering}
\end{figure}

\begin{figure}[t!]
\begin{centering}
\includegraphics[width=\figsize]{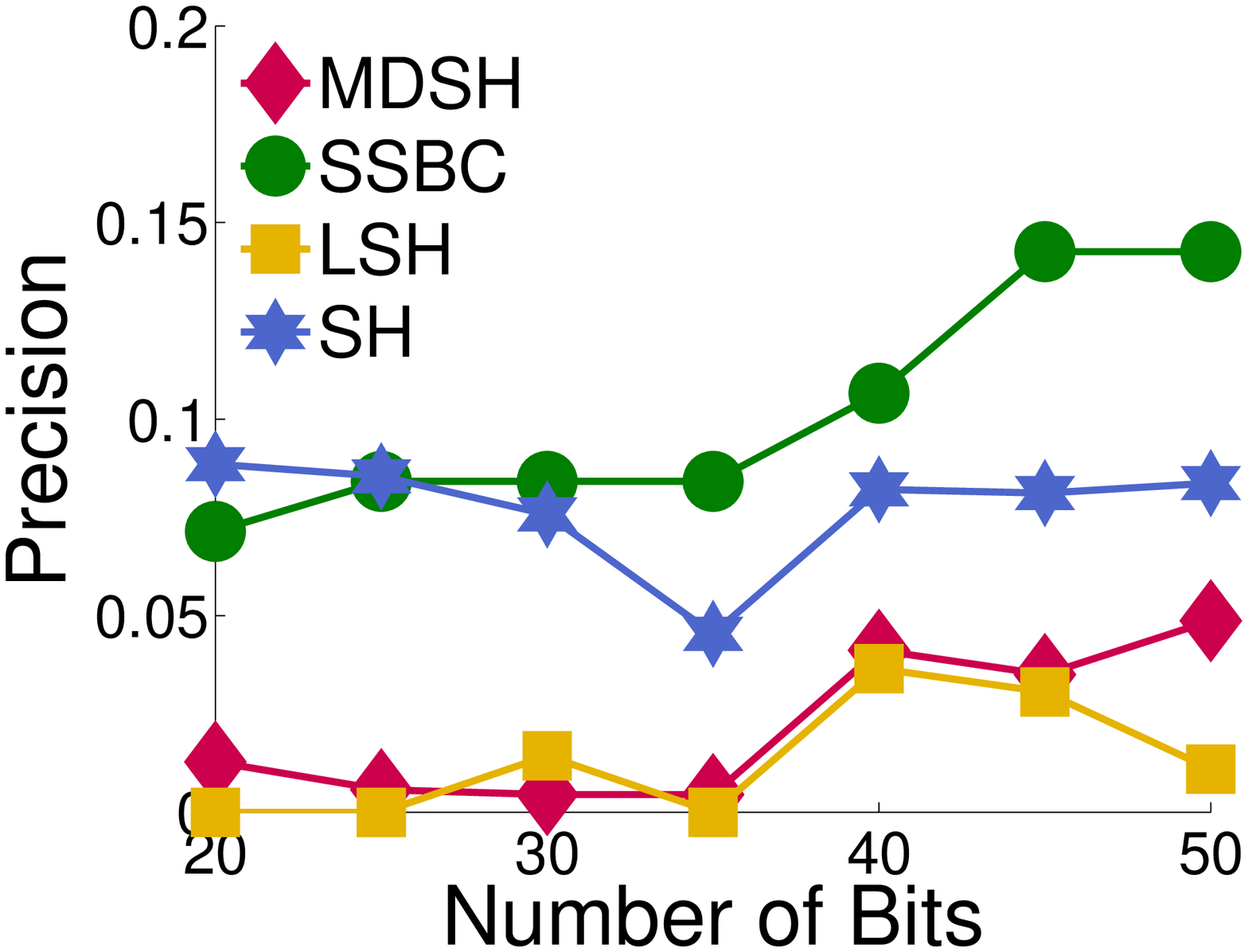}
\includegraphics[width=\figsize]{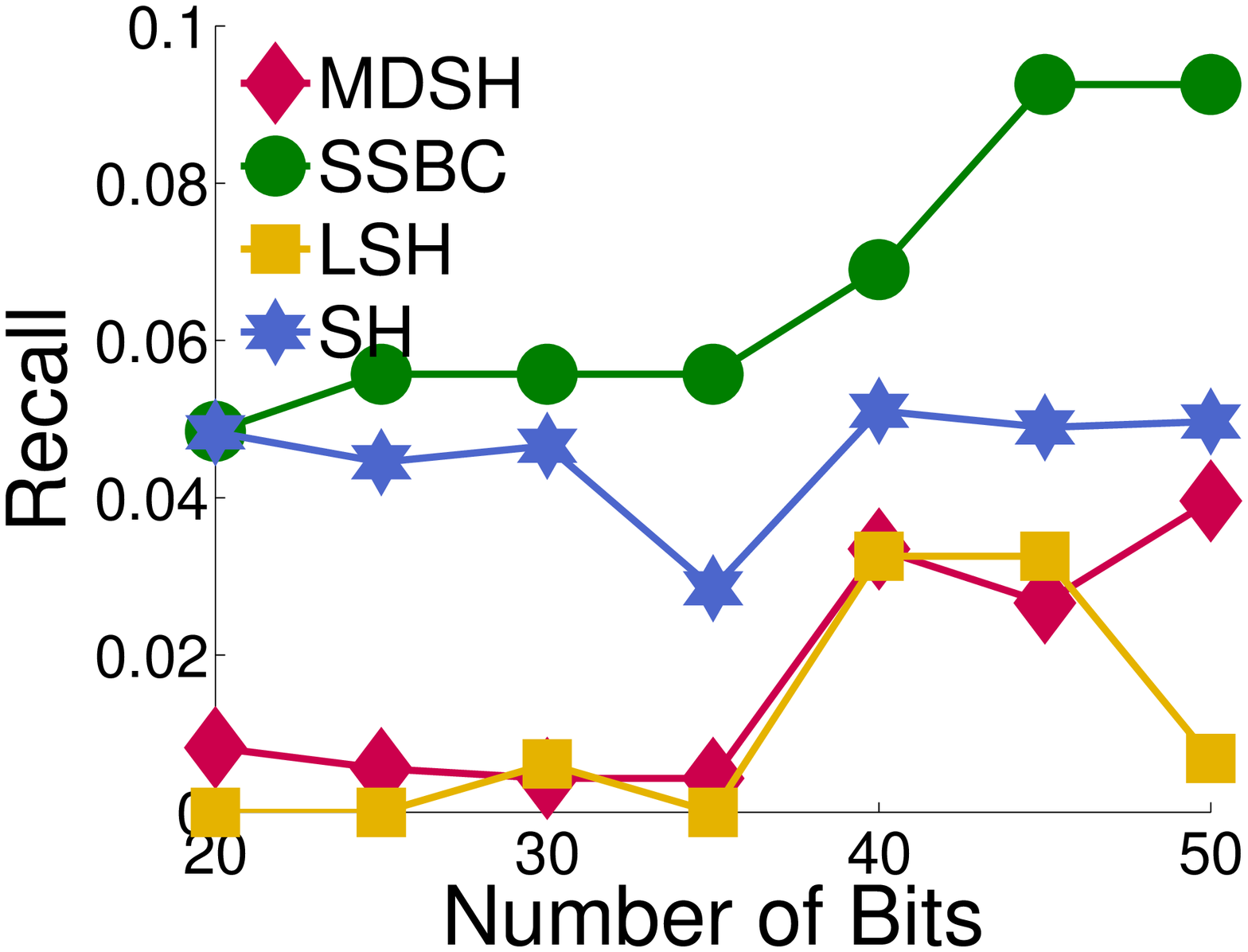}
\includegraphics[width=\figsize]{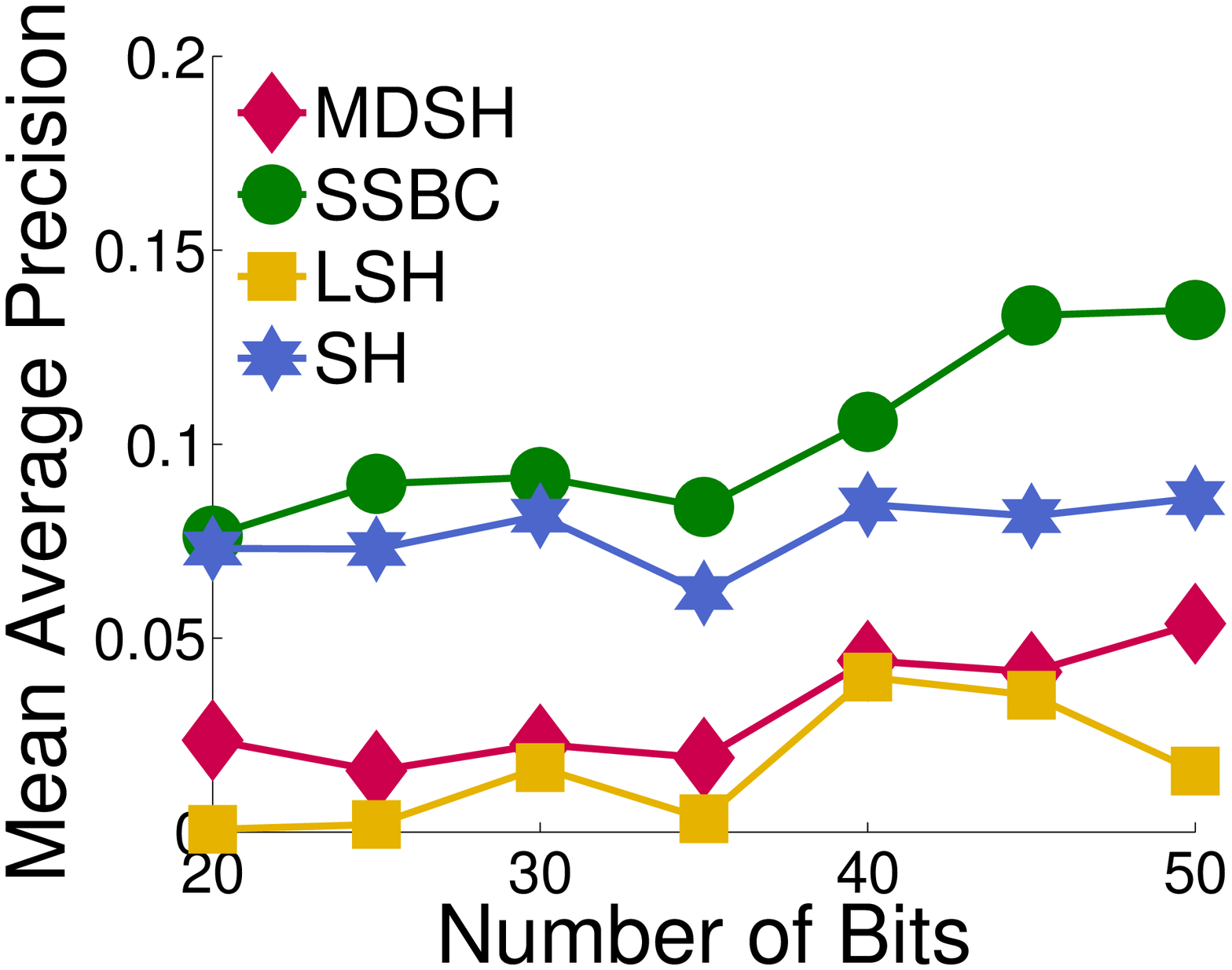}
\includegraphics[width=\figsize]{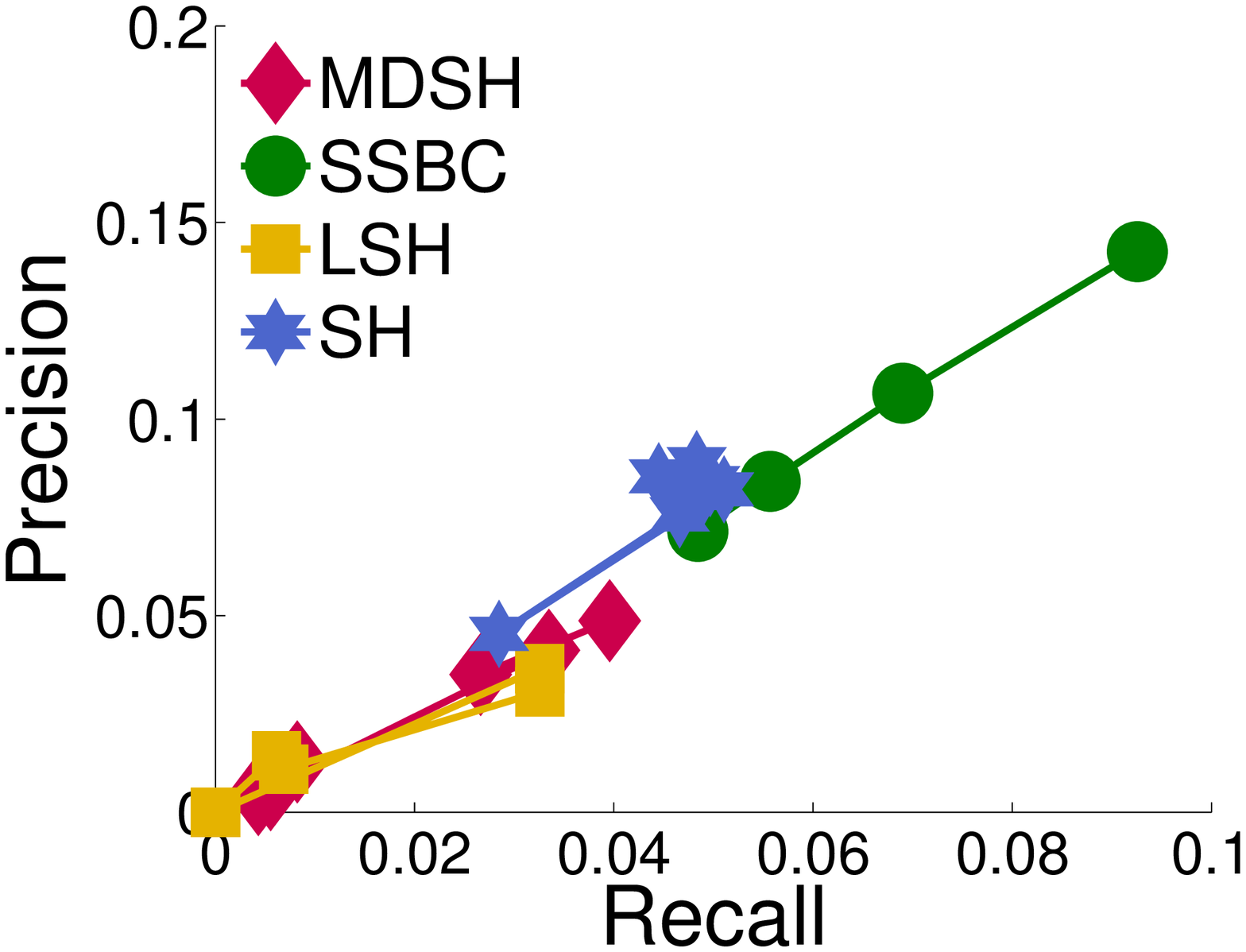}
\vspace{-3mm}
\caption{\label{fig:cbm}
Results on \s{CBM} dataset.}  
\vspace{-2mm}
\end{centering}
\end{figure}

\begin{figure}[t!]
\begin{centering}
\includegraphics[width=\figsize]{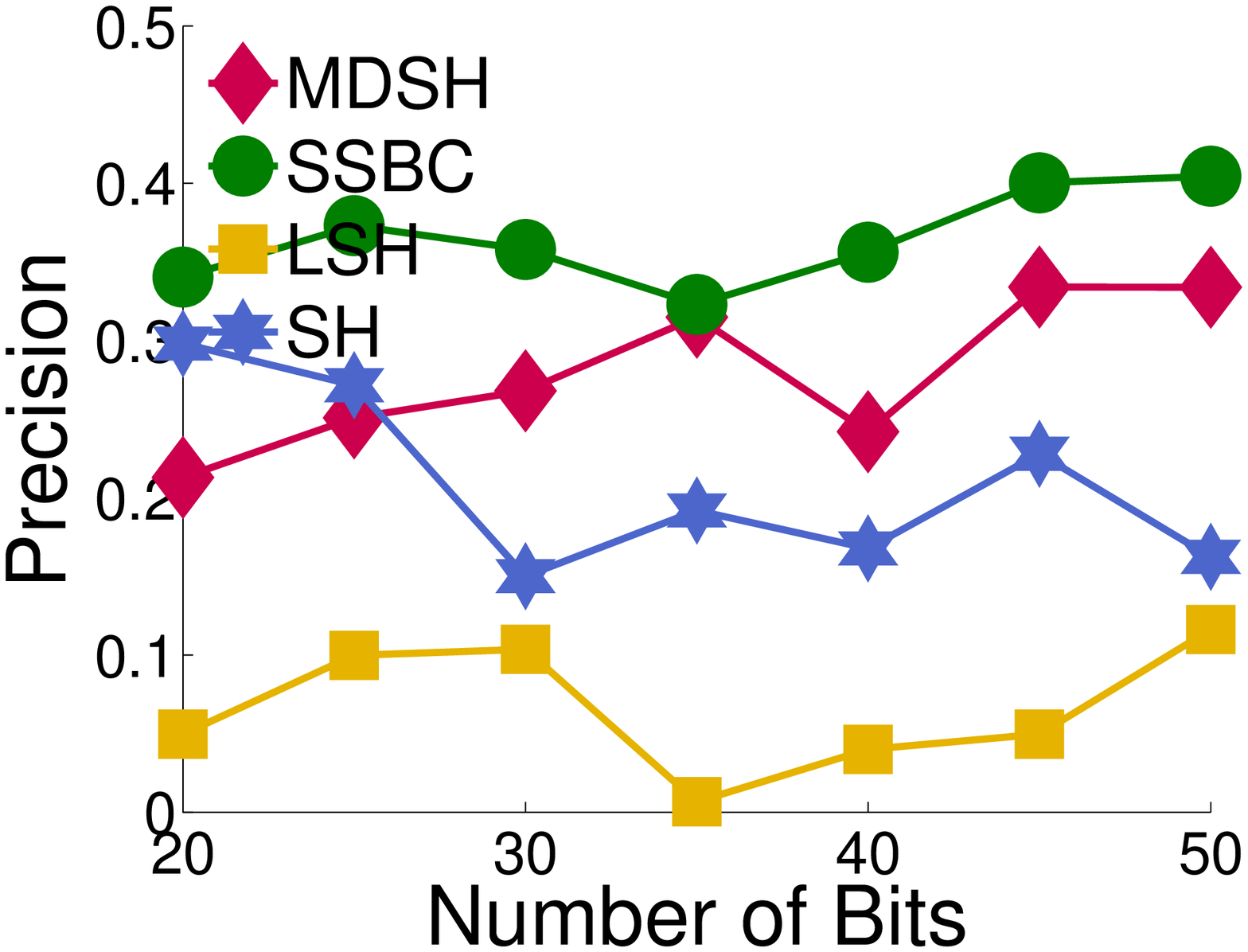}
\includegraphics[width=\figsize]{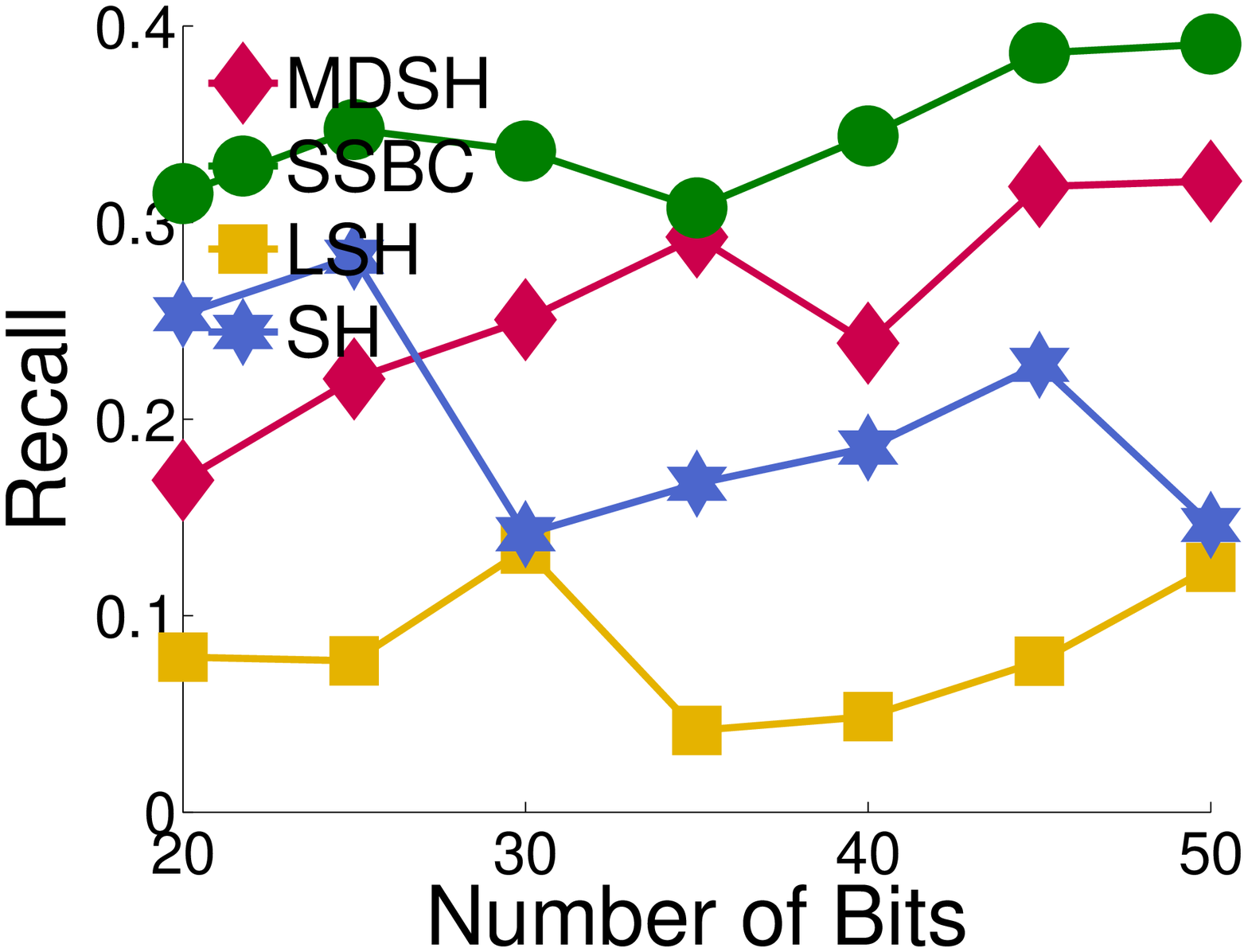}
\includegraphics[width=\figsize]{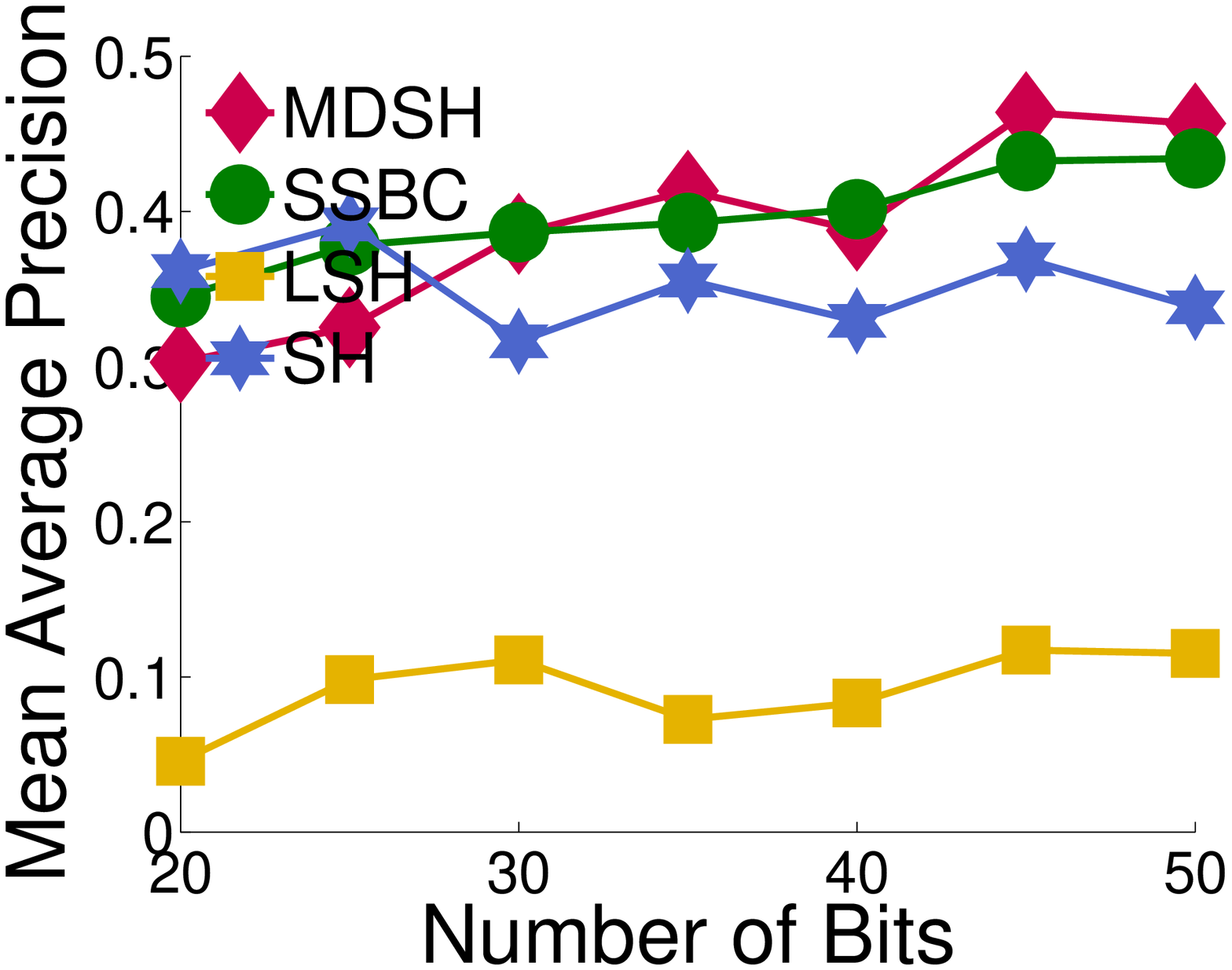}
\includegraphics[width=\figsize]{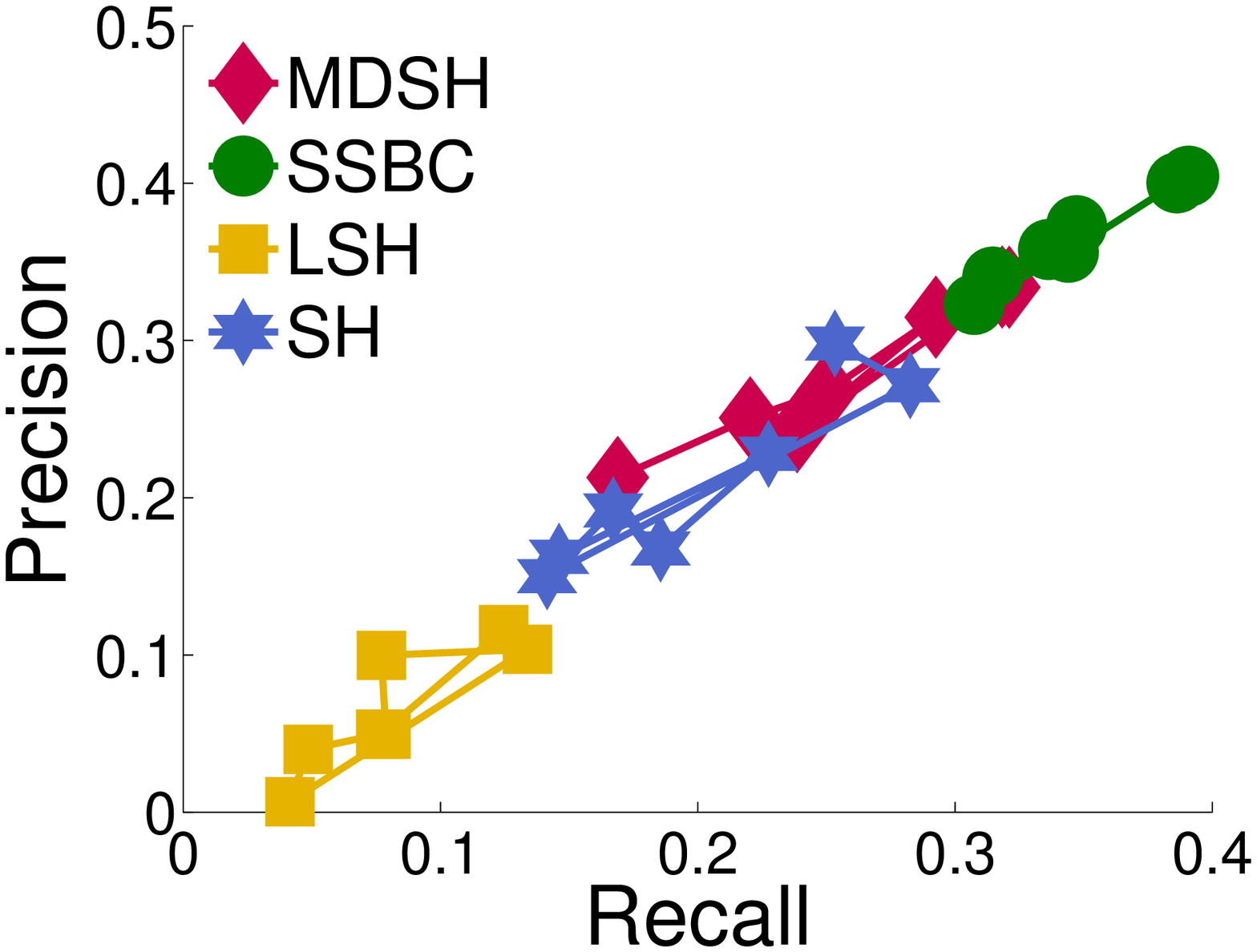}
\vspace{-3mm}
\caption{\label{fig:uniform}
Results on \s{Uniform} dataset.}  
\vspace{-2mm}
\end{centering}
\end{figure}

\begin{figure}[t!]
\begin{centering}
\includegraphics[width=\figsize]{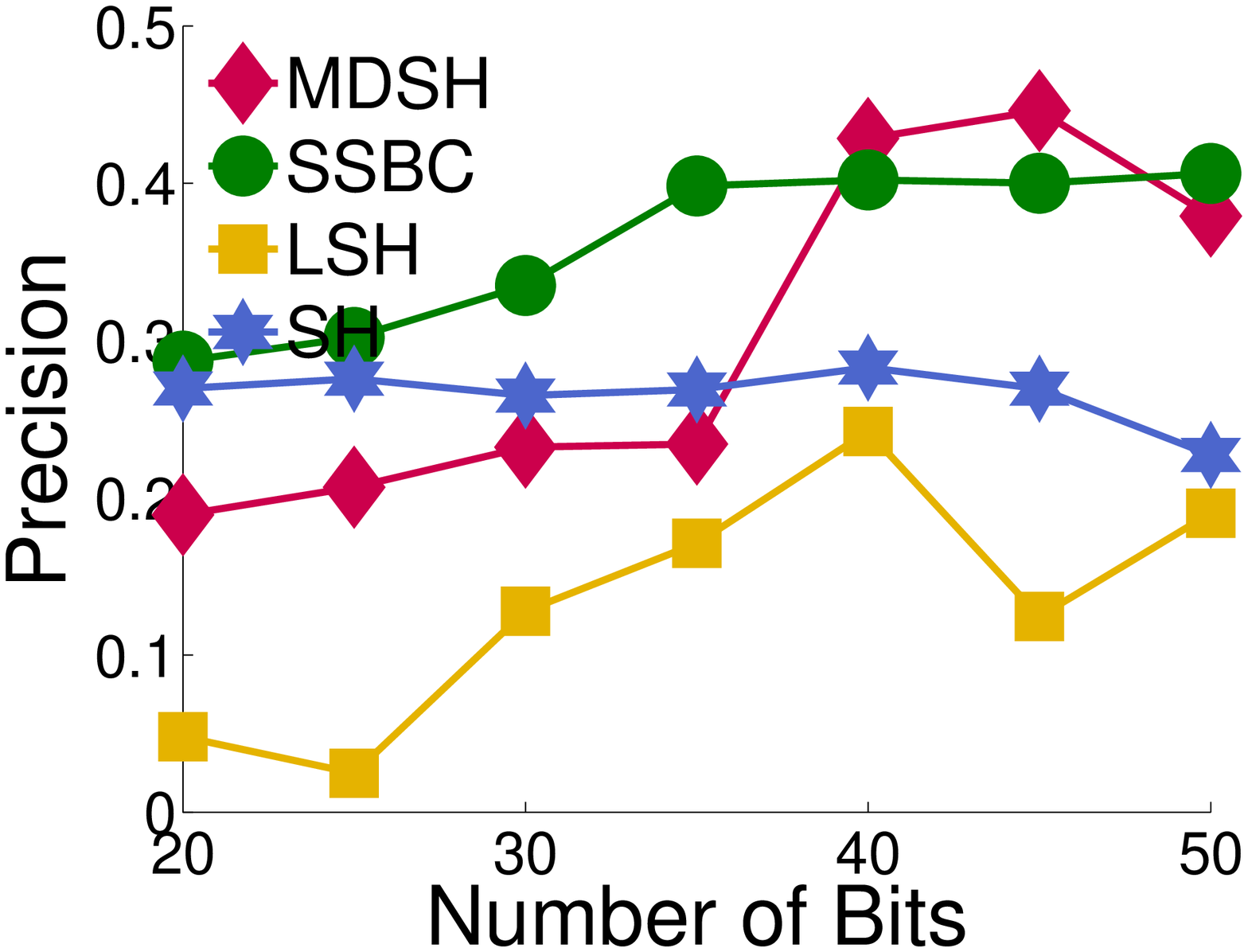}
\includegraphics[width=\figsize]{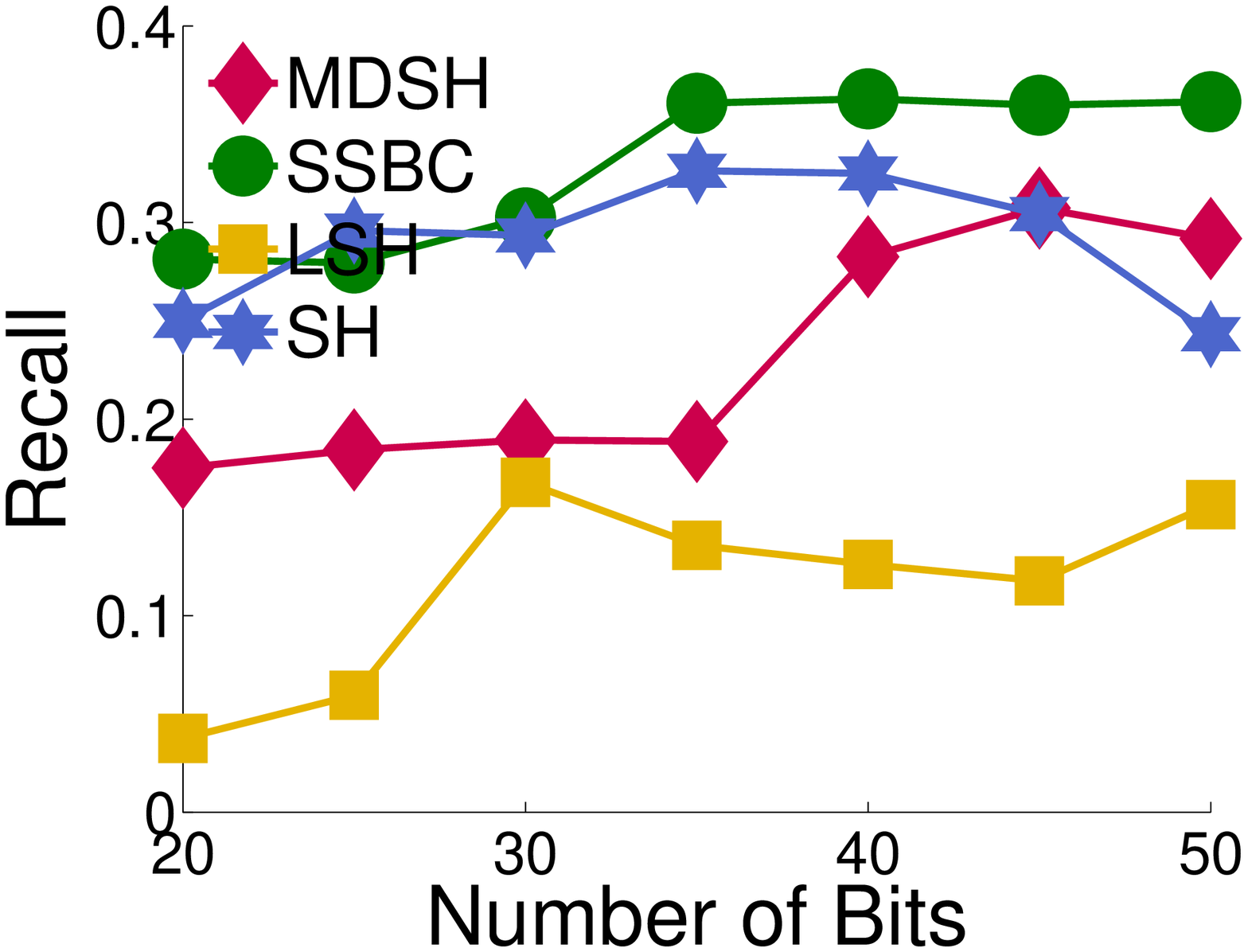}
\includegraphics[width=\figsize]{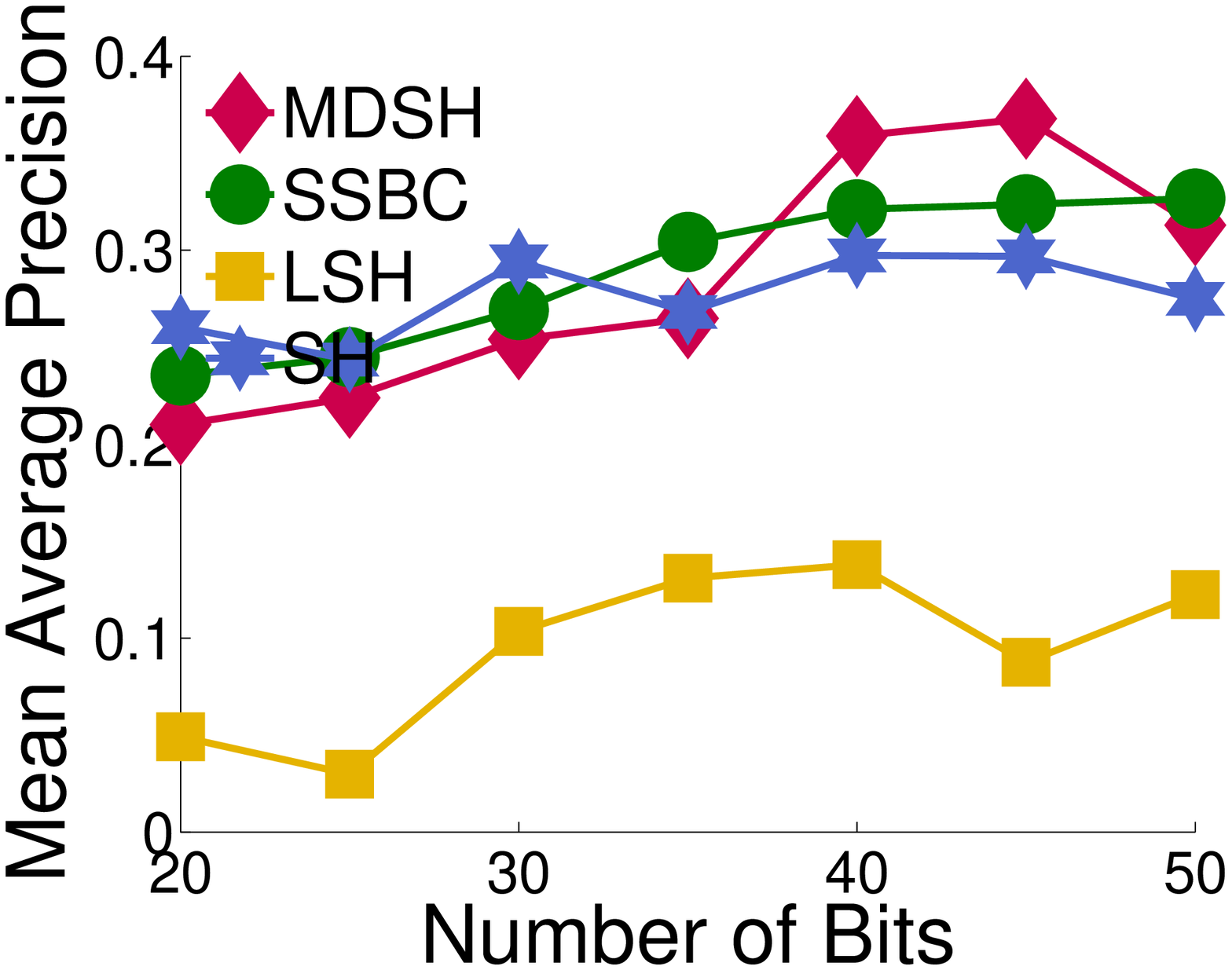}
\includegraphics[width=\figsize]{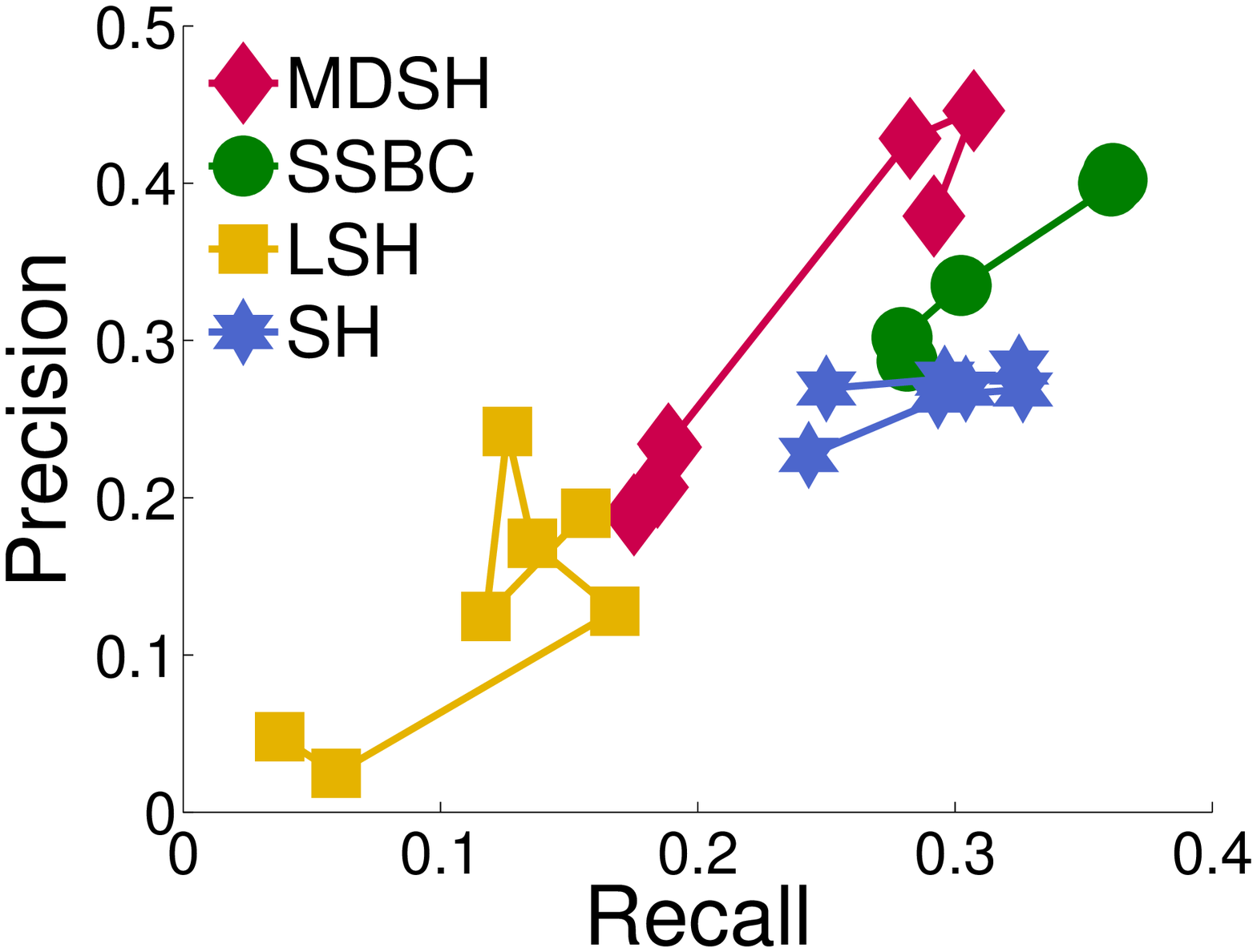}
vspace{-3mm}
\caption{\label{fig:covtype}
Results on \s{Covtype} dataset.}  
\vspace{-2mm}
\end{centering}
\end{figure}

\begin{figure}[t!]
\begin{centering}
\includegraphics[width=\figsize]{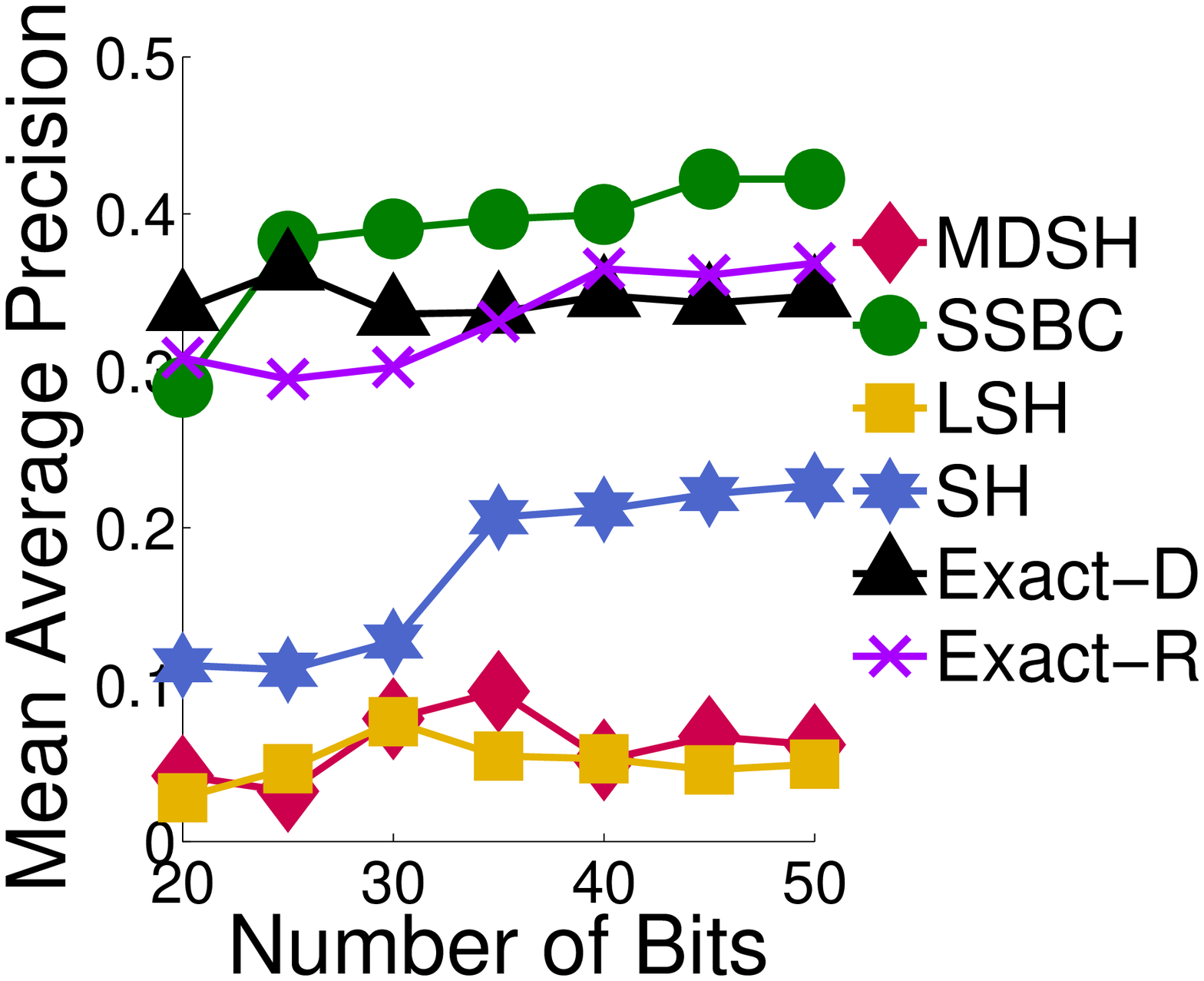}
\includegraphics[width=\figsize]{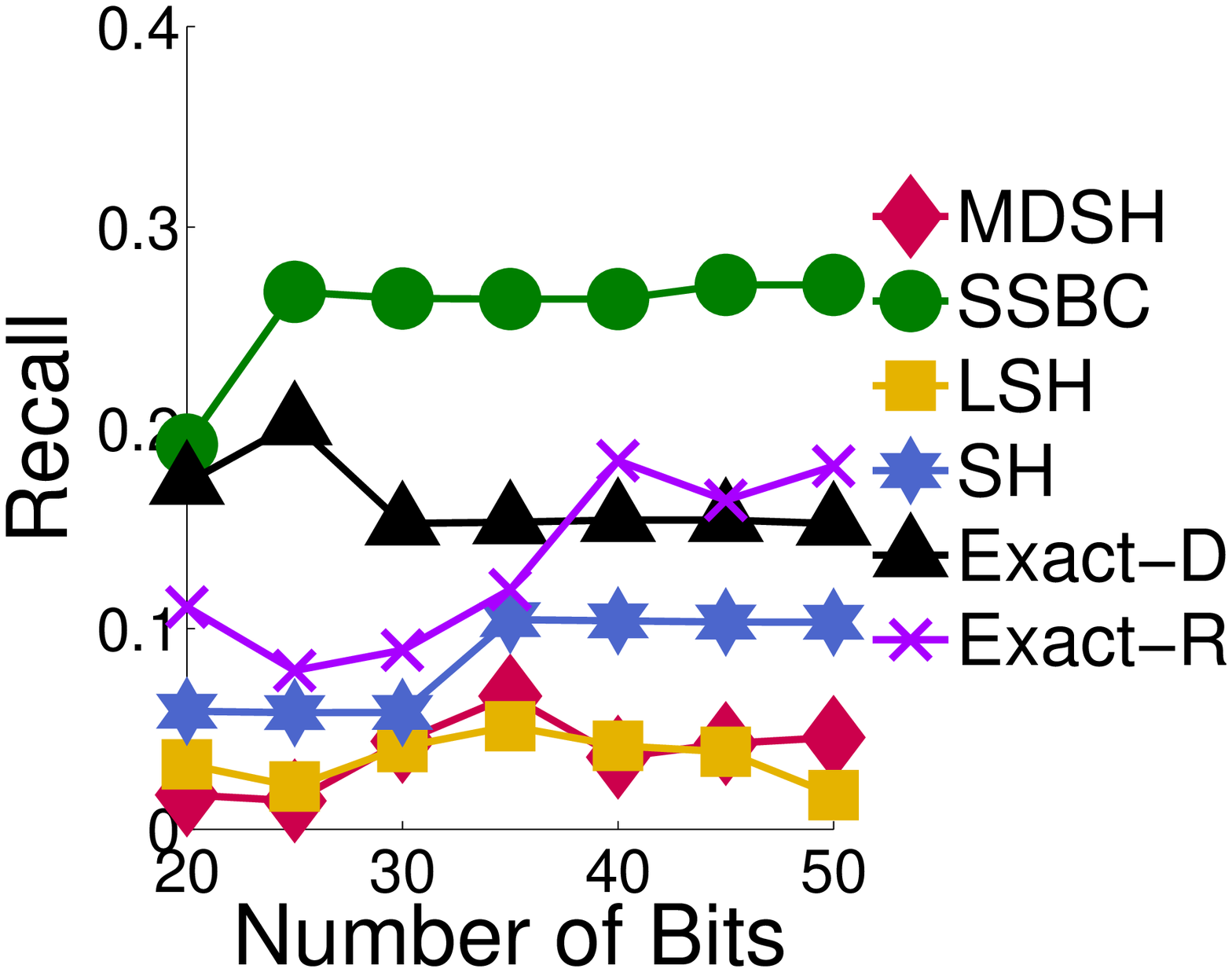}
\includegraphics[width=\figsize]{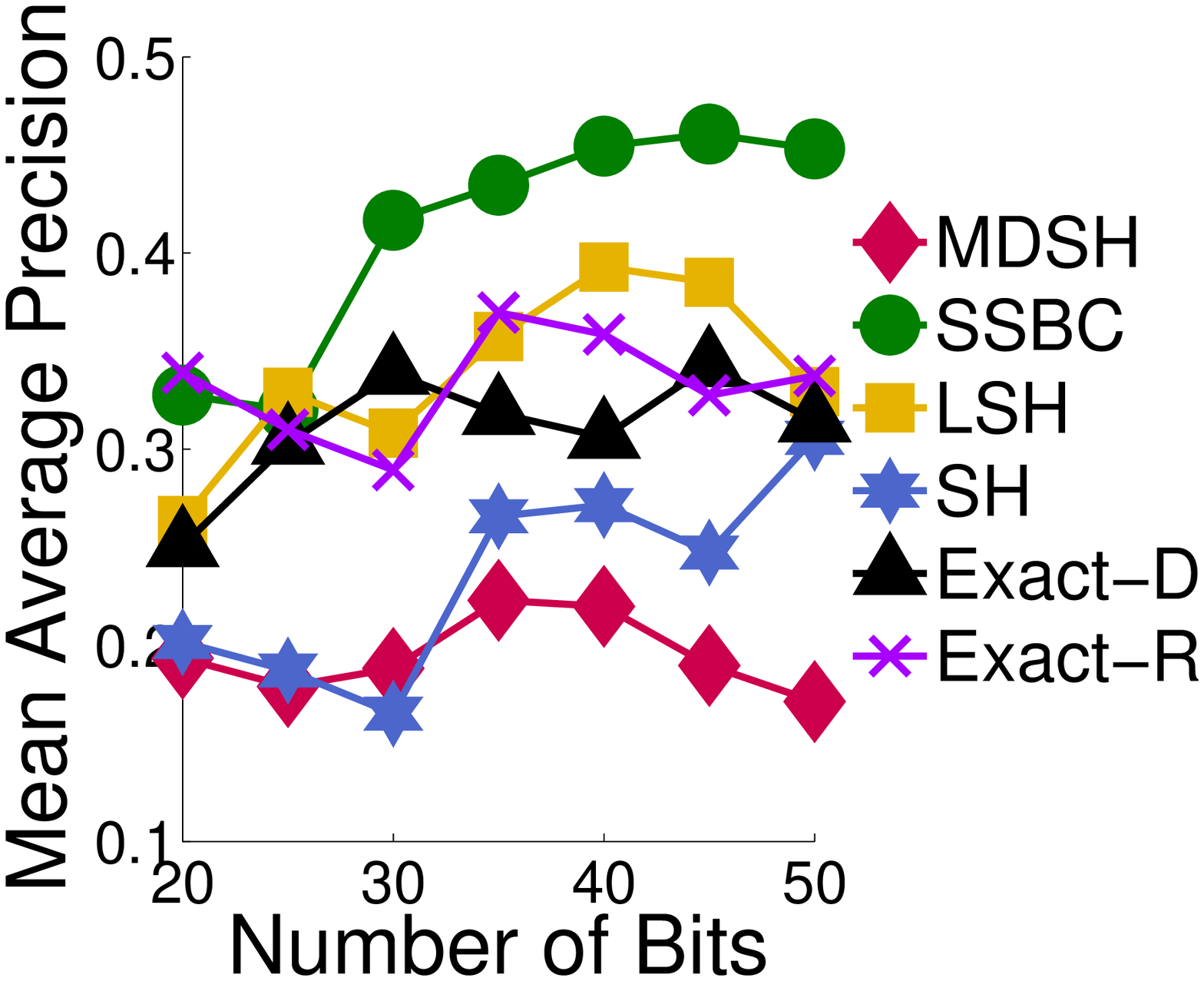}
\includegraphics[width=\figsize]{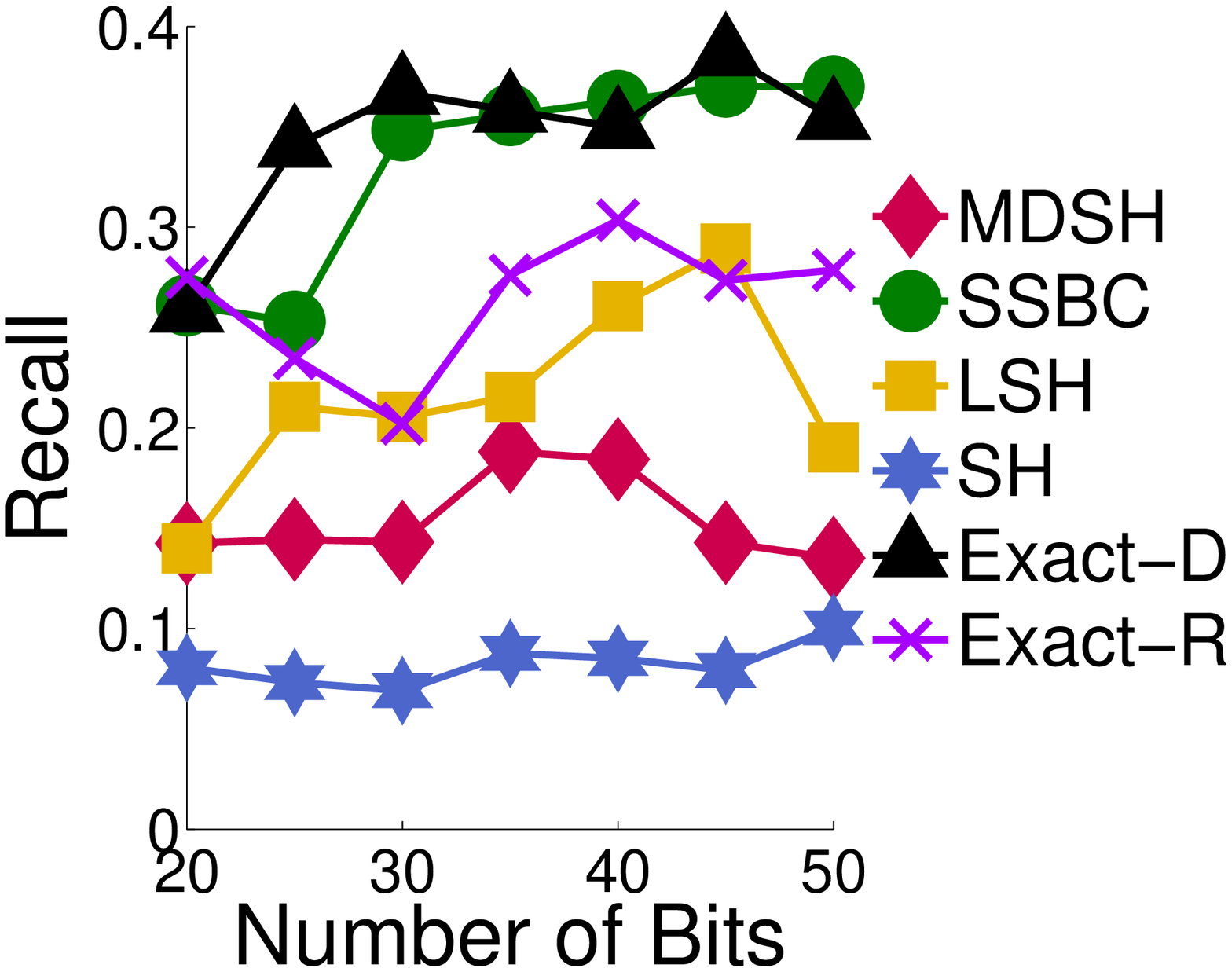}
\vspace{-3mm}
\caption{\label{fig:exact_cbm_pamap}
Comparing algorithms with Exact methods on \s{CBM} (first row) and \s{PAMAP} (second row) datasets. Training set for each dataset was of size $200$ and $100$ and test set was of size $1000$ and $3000$, respectively.}  
\vspace{-2mm}
\end{centering}
\end{figure}

As we observe in precision and recall plots of figures \ref{fig:uniform},\ref{fig:covtype},\ref{fig:cbm} and \ref{fig:pamap}, SSBC performs exceptionally well on precision, providing very few ``false positives" compared to the other algorithms and consistently providing the highest precision of the methods evaluated. 
On recall metric also SSBC provides the best results over all the approaches evaluated. In both cases, this edge in performance is maintained over all tested ranges of length $k \in [20,50]$ of codewords. Combining these two plots we get precision-recall comparison (last plot in all above mentioned figures) which shows that SSBC forms an almost 45-degree line in all figures, i.e. basically its mistake rate does not increase by returning more candidates for nearest neighbours (having high recall).

In a separate set of experiments, we compared accuracy of all algorithms with exact methods. This time in order to allow exact algorithms to load the whole $n$ by $n$ weight matrix in RAM, we used a much smaller test set. Size of test set and training set for these experiments are mentioned in caption of plot \ref{fig:exact_cbm_pamap}.
As we see in this plot, SSBC secures higher mean average precision and recall than the exact methods, ``exact-D'' and ``exact-R'' which solve the matrix optimization by applying an SVD over enitre dataset. This is likely because maintaining a column sample of the weight matrix through a training set helps prevent overfitting errors.

\bibliographystyle{abbrv}
\bibliography{hashing}  

\end{document}